\documentclass[a4paper,11pt]{article}
\usepackage[left=.7in,right=.7in,top=.9in,bottom=1.2in]{geometry}
\usepackage[T1]{fontenc}
\usepackage[utf8]{inputenc}
\usepackage{lmodern}
\usepackage{mathpazo}
\usepackage[nomath]{kpfonts}
\usepackage{xcolor}
\usepackage{graphicx}
\usepackage{anyfontsize}
\usepackage{natbib}
\usepackage{caption}
\usepackage{subcaption}
\usepackage{amsmath,amssymb,bm}
\makeatletter
\def\newrefformat#1#2{%
  \@namedef{pr@#1}##1{#2}}
\newrefformat{eq}{\textup{(\ref{#1})}}
\newrefformat{lem}{Lemma \ref{#1}}
\newrefformat{thm}{Theorem \ref{#1}}
\newrefformat{cha}{Chapter \ref{#1}}
\newrefformat{sec}{Section \ref{#1}}
\newrefformat{tab}{Table \ref{#1} on page \pageref{#1}}
\newrefformat{fig}{Figure \ref{#1} on page \pageref{#1}}
\def\prettyref#1{\@prettyref#1:}
\def\@prettyref#1:#2:{%
  \expandafter\ifx\csname pr@#1\endcsname\relax%
    \PackageWarning{prettyref}{Reference format #1\space undefined}%
    \ref{#1:#2}%
  \else%
    \csname pr@#1\endcsname{#1:#2}%
  \fi%
}
\makeatother
\newrefformat{eq}     {(\ref{#1})}
\newrefformat{fig}    {{Figure~\ref{#1}}}
\newrefformat{tab}    {{Table~\ref{#1}}}
\newrefformat{prop}   {{Proposition~\ref{#1}}}
\newrefformat{lem}   {{Lemma~\ref{#1}}}
\newrefformat{def}    {{Definition~\ref{#1}}}
\newrefformat{app}    {{Appendix~\ref{#1}}}
\newrefformat{foot}   {{Footnote~\ref{#1}}}
\newrefformat{assum}  {{{Assumption~\ref{#1}}}}
\newrefformat{sub}    {{{Section~\ref{#1}}}}
\newrefformat{def}    {{{Definition~\ref{#1}}}}
\newrefformat{fact}   {{Fact~\ref{#1}}}
\newrefformat{enum}   {{\ref{#1}}}
\newrefformat{remark} {{Remark~\ref{#1}}}
\newrefformat{cor}    {{Corollary~\ref{#1}}}
\newrefformat{property} {{Property~\ref{#1}}}
\newrefformat{obs} {{Observation~\ref{#1}}}
\newrefformat{conjecture} {{Conjecture~\ref{#1}}}
\newrefformat{claim} {{Claim~\ref{#1}}}
\usepackage{amsthm}
\newtheorem {proposition}{Proposition}	%

\newtheorem {lemma}      {Lemma}	    %

\newtheorem {fact}       {Fact}	
\newtheorem {assumption} {Assumption}
\theoremstyle{definition}
\newtheorem*{definition*}{Definition}

\newtheorem*{theorem*}{Theorem}
\newtheorem{remark}{Remark}      		%

\theoremstyle{plain}

\usepackage{etoolbox}
\newcommand{\addQEDstyle}[2]{%
\AtBeginEnvironment{#1}{\pushQED{\qed}\renewcommand{\qedsymbol}{#2}}%
\AtEndEnvironment{#1}{\popQED}}
\addQEDstyle{example}{$\lozenge$}
\addQEDstyle{observation}{$\lozenge$}
\addQEDstyle{remark}{$\lozenge$}
\addQEDstyle{aobservation}{$\lozenge$}
\addQEDstyle{aremark}{$\lozenge$}
\addQEDstyle{fact}{$\lozenge$}
\AfterEndEnvironment{definition}{\noindent\ignorespaces}
\AfterEndEnvironment{proposition}{\noindent\ignorespaces}
\AfterEndEnvironment{lemma}{\noindent\ignorespaces}
\AfterEndEnvironment{proof}{\noindent\ignorespaces}
\usepackage{enumitem}
\setlist[enumerate]{label=(\alph*),itemsep=.2em,topsep=.3em}
\makeatletter

\newcommand{\diag}[1]{\operatornamewithlimits{diag}[#1]}
\newcommand\PDF[2]{\dfrac{\partial{#1}}{\partial{#2}}}

\newenvironment{figurenotes}[1][Notes]%
    {\vskip .3em\begin{minipage}[t]{\linewidth}\footnotesize{\itshape#1: }}
    {\end{minipage}}

\newcommand{\argmax}{\operatornamewithlimits{arg\,max}}

\makeatother
\allowdisplaybreaks%
\usepackage{overpic}
\usepackage{mathtools}
\usepackage{mathrsfs} 
\usepackage{arydshln}
\usepackage{multirow}
\usepackage{setspace}
\usepackage[bottom]{footmisc}
\usepackage[bookmarks=true,colorlinks=true,linkcolor=blue,urlcolor=blue,citecolor=blue,breaklinks=true,setpagesize=false,linktocpage]{hyperref}

\definecolor{RED}{rgb}{1,.05,0}

\usepackage[misc]{ifsym}

\providecommand\SYM{\multicolumn{2}{c}{\text{\small\emph{Sym.}}}}

\numberwithin{equation}{section}
\usepackage{amsmath,amssymb,bm}

\providecommand{\PDF}[2]{\frac{\partial{#1}}{\partial{#2}}}

\providecommand{\Vt}[1]{\bm{#1}}
\providecommand{\Vta}{\bm{a}}

\providecommand{\Vte}{\bm{e}}
\providecommand{\Vtf}{\bm{f}}

\providecommand{\Vtv}{\bm{v}}
\providecommand{\Vtw}{\bm{w}}
\providecommand{\Vtx}{\bm{x}}
\providecommand{\Vty}{\bm{y}}
\providecommand{\Vtz}{\bm{z}}
\providecommand{\VtA}{\mathbf{A}}

\providecommand{\VtD}{\mathbf{D}}

\providecommand{\VtF}{\mathbf{F}}
\providecommand{\VtG}{\mathbf{G}}

\providecommand{\VtI}{\mathbf{I}}

\providecommand{\VtM}{\mathbf{M}}

\providecommand{\VtP}{\mathbf{P}}

\providecommand{\VtV}{\mathbf{V}}
\providecommand{\VtW}{\mathbf{W}}

\providecommand{\VtZ}{\mathbf{Z}}

\providecommand{\VtPsi}{\bm{\Psi}}

\providecommand{\Rmi}{\mathrm{i}}

\providecommand{\RmC}{\mathrm{C}}

\providecommand{\ClK}{\mathcal{K}}

\providecommand{\ClN}{\mathcal{N}}

\providecommand{\ClX}{\mathcal{X}}
\providecommand{\ClY}{\mathcal{Y}}

\providecommand{\BbR}{\mathbb{R}}

\providecommand{\Htphi}{\hat{\phi}}

\providecommand{\Bra}{\bar{a}}

\providecommand{\Brm}{\bar{m}}

\providecommand{\Brv}{\bar{v}}
\providecommand{\Brw}{\bar{w}}
\providecommand{\Brx}{\bar{x}}

\providecommand{\Bromega}{\bar{\omega}}

\providecommand{\Tlw}{\tilde{w}}

\providecommand{\TlU}{\tilde{U}}

\providecommand{\Tlomega}{\tilde{\omega}}

\providecommand{\Dty}{\dot{y}}

\providecommand{\Is}{\equiv}

\providecommand{\HtVta}{\hat{\Vta}}

\providecommand{\HtVtv}{\hat{\Vtv}}
\providecommand{\HtVtw}{\hat{\Vtw}}
\providecommand{\HtVtx}{\hat{\Vtx}}
\providecommand{\HtVty}{\hat{\Vty}}

\providecommand{\BrVtx}{\bar{\Vtx}}

\providecommand{\BrVtD}{\bar{\VtD}}

\providecommand{\BrVtG}{\bar{\VtG}}

\providecommand{\TlVtf}{\tilde{\Vtf}}

\providecommand{\TlVtw}{\tilde{\Vtw}}

\providecommand{\DtVtx}{\dot{\Vtx}}

\providecommand{\inK}{\in\ClK}

\providecommand{\inN}{\in\ClN}

\providecommand{\inX}{\in\ClX}
\providecommand{\inY}{\in\ClY}

\usepackage{times}
\begin{document}
\title{Production Externalities and Dispersion Process in a Multi-region Economy%
\thanks{We wish to thank Kristian Behrens, Shota Fujishima, Noriaki Matsushima, Tomoya Mori, Yuki Takayama, and Dao-Zhi Zeng for useful comments. Minoru Osawa thanks grant support from JSPS Kakenhi 17H00987, 18K04380, and 19K15108. Jos\'e M. Gaspar gratefully acknowledges financial support from Funda\c{c}\~{a}o para a Ci\^{e}ncia e Tecnologia (through projects  UIDB/04105/2020, UIDB/00731/2020, PTDC/EGE-ECO/30080/2017 and CEECIND/02741/2017).  Part of this research was conducted under the project ``Research on the evaluation of spatial economic impacts of building bus termini (Principal Investigator: Prof. Yuki Takayama, Kanazawa University)'', supported by the Committee on Advanced Road Technology under the authority of the MLIT in Japan.}}
\author{Minoru Osawa\thanks{Corresponding author. osawa.minoru.4z@kyoto-u.ac.jp, Institute of Economic Research, Kyoto University}\qquad\quad Jos\'e M. Gaspar\thanks{jgaspar@porto.ucp.pt, CEGE and Cat\'olica Porto Business School, Universidade Cat\'olica Portuguesa}}
\date{\small\today}
\maketitle
\begin{abstract} 

We consider an economic geography model with two inter-regional proximity structures: one governing goods trade and the other governing production externalities across regions.  
We investigate how the introduction of the latter affects the timing of endogenous agglomeration and the spatial distribution of workers across regions. 
As transportation costs decline, the economy undergoes a progressive dispersion process. 
Mono-centric agglomeration emerges when inter-regional trade and/or production externalities incur high transportation costs, while uniform dispersion occurs when these costs become negligibly small (i.e., when distance dies). 
In a multi-regional geography, the network structure of production externalities can determine the geographical distribution of workers as economic integration increases. 
If production externalities are governed solely by geographical distance, a mono-centric spatial distribution emerges in the form of suburbanization. 
However, if geographically distant pairs of regions are connected through tight production linkages, multi-centric spatial distribution can be sustainable.
\end{abstract}



\clearpage

\section{Introduction}
\label{sec:introduction}

The spectacular drop in transportation costs due to the Industrial Revolution has led to the concentration of economic activities and population in fewer and fewer geographical locations. 
In this context, spatial economic theory has emphasized the roles of endogenous forces in shaping lasting and sizable economic agglomerations in the modern economy. 
Agglomeration economies are considered central to the formation of major economic clusters because the concentration of economic actors in cities produce diverse positive effects, both pecuniary and non-pecuniary \citep{Duranton-Puga-HB2004,Duranton-Puga-HB2015}. 
The spatial economy can be seen as the result of trade-offs between such scale economies and the transportation costs incurred by the movement of goods, people, and information \citep{Proost-Thisse-JEL2019}. 
However, the concentration of production in a small number of selected locations is currently in the process of evening out, partly due to the significant developments in information and communication technologies, which have made it possible to organize complex production processes even when they are separated by geographical distance \citep{Baldwin-Book2016}.
The ``death of distance'' \citep{cairncross1997death}, which indicates that communication technology obviates the need for physical proximity, seems to have become increasingly relevant over the past two decades. 
To understand the process of ``re-dispersion'' from a spatially localized economy, we need theories to study the interaction between multiple spatial linkages, in contrast to the conventional economic geography models that focus mainly on trade linkages.

We study how the structure of production externalities between multiple regions affects the unbundling process of a geographically mono-centric economic agglomeration when inter-regional trade costs and/or communication costs become negligibly small (i.e., when distance dies). 
Further, to obtain first-order theoretical implications, we develop a bare-bones economic geography model with one differentiated sector and a finite number of locations. 
The trade of goods is costly and there is inter-regional productivity spillover, where the spatial concentration of workers leads to positive effects. 
Based on the model, we draw qualitative insights into the \textit{timing} and \textit{forms} of workers' agglomeration and dispersion across regions.

In urban economics, there is an important literature stream on the economics of agglomeration  \citep[e.g.,][]{Beckmann-Book1976,Fujita-Ogawa-RSUE1982,Lucas-Rossi-Hansberg-ECTA2002,Helsley-Strange-JPE2014} explaining how inter-location externalities influence the urban spatial structure, including the location of firms and households \citep[for a survey, see][]{Fujita-Thisse-Book2013,Duranton-Puga-HB2004,Duranton-Puga-HB2015}. 
Such urban externalities may be microfounded by the interesting literature of inter-individual production externalities and cities  \citep[e.g.,][]{Helsley-Zenou-JET2014,Picard-Zenou-JUE2018}. 
We take an intermediate strategy between the studies on individual-level production externalities and geographical proximity as \textit{the} determinant of agglomeration economy. 

We consider a simple general equilibrium framework and 
we assume a perfectly competitive Armington model with positive externalities and iceberg transportation costs. 
Productivity of a region depends on the entire spatial distribution of workers, rather than the local spillovers within each region. 
As per \cite{Rosenthal-Strange-JEP2020}'s survey, agglomeration effects can act at various spatial scales (e.g., regional, metropolitan, and neighborhood scales, or even within each building). 
To represent such effects flexibly, we assume that a region's productivity depends on an \textit{externality matrix}, $\VtG = [\psi_{ij}]$, which represents the structure of inter-regional productivity spillovers. 
Each $\psi_{ij} \in (0,1]$ represents the level of positive externality a worker in region $j$ has on region $i$'s productivity $a_i$, so that $a_i = \sum_{j} \psi_{ij} x_j$, where $x_j$ denotes the mass of workers in region $j$. 
The externality matrix can be interpreted either as industrial or knowledge linkages. 
This setup includes, as a special case, spatially decaying technological externalities in urban models if, for example, $\psi_{ij} = \exp\left(-\tau \ell_{ij}\right)$ with a distance decay parameter $\tau > 0$ and $\ell_{ij}$ is the geographical distance  between locations $i$ and $j$ \citep[e.g.,][]{Fujita-Ogawa-RSUE1982,Ahlfeldt-et-al-ECTA2015}. 
Our flexible specification allows us to illustrate the roles played by the underlying geographical proximity structure and by the network structure of production externalities.  

We first consider a symmetric two-region economy to elucidate the main workings of the model. 
There are two key transportation cost parameters: 
the freeness of trade between regions and the level of productivity spillover between region (i.e., the spatial extent of production externalities). 
These parameters form the index of transportation costs. 
Endogenous agglomeration emerges when either parameter is small. 
Symmetric dispersion of workers occurs when trade friction is low (i.e., when distance dies). 
When transportation is prohibitively costly, there is full agglomeration of workers in one region.  
As transportation costs decline, this asymmetry is gradually resolved:  
the spatial configuration becomes increasingly symmetric, leading to the symmetric dispersion at some threshold values of transportation costs. 
This behavior is akin to the dispersion process in economic geography models with urban costs \citep[e.g.,][]{Helpman-Book1998,Tabuchi-JUE1998} in that dispersion occurs when trade is freer, although the dispersion forces at work in our model stem from production externalities rather than land rents.
The role of production externalities in our model is intuitive. 
When the spatial extent of production externalities is large (i.e., when producing in a smaller region is not disadvantageous because there are no significant productivity differences between the two regions) economic activities tend to become more dispersed. 

We then explore a symmetric four-region economy and consider various structures for the externality matrix. 
The four-region configuration is the minimal setting to investigate the roles of the \textit{network structure} of $\VtG$,  since non-trivial (but symmetric) network structures can emerge only when the number of regions is more than three. 
Similar to the two-region setup, one region attracts almost all workers when trade costs are low and inter-regional production externalities spillover is negligible. 
The economy then exhibits a gradual dispersion process as the transportation costs in these channels decline. 
The agglomeration force tends to support a geographically mono-centric pattern of workers along the dispersion process. 
If the economy is more integrated as a whole in terms of the network structure of $\VtG$, then agglomeration is less likely, compared with less integrated networks. 
This is because the endogenous advantage due to production spillover plays a less prominent role when the economy is more integrated with respect to $\VtG$. 
Additionally, if some pair of regions is ``closer'' with respect to $\VtG$, then geographical configurations other than the mono-centric can be sustainable. 
For instance, a duo-centric concentration of workers emerges if geographically distant pairs of regions are closer in terms of $\VtG$ because of, for example, passenger transportation modes with economy of distance (e.g., regional airlines). 
In sum, the network structure of production externalities can determine when dispersion is attained and how it looks like in the physical space during the process of unbundling of a mono-centric economic agglomeration. 

The remainder of this paper is organized as follows: 
\prettyref{sec:literature} discusses the related literature; 
\prettyref{sec:model} formulates the model; 
\prettyref{sec:two-region_economy} studies the model under the simplest possible setup, the symmetric two-region economy; 
\prettyref{sec:four-region_economies} illustrates the fundamental roles of the structure of the interaction network $\VtG$ and the effects of variation in inter-regional proximity structure $\VtG$ employing stylized examples; 
finally, \prettyref{sec:conclusion} concludes the paper. 
All proofs and technical discussions are presented in \prettyref{app:proofs}. 

\section{Related literature}
\label{sec:literature}

The current unbundling process of economic agglomeration can also be explained by the ``bell-shaped development'' narrative for industrial agglomeration  \citep[][Section 8]{Fujita-Thisse-Book2013}. 
The seminal theory of endogenous regional agglomeration of \cite{Krugman-JPE1991} predicts that the spatial distribution of economic activities in a country is organized into a mono-centric state when transportation costs decline below a threshold in a multi-region economy \citep{Tabuchi-Thisse-JUE2011,ikeda2012spatial,Akamatsu-Takayama-Ikeda-JEDC2012,AMOT-DP2019}. 
This prediction from the theory of endogenous agglomeration is qualitatively consistent with the evolutionary trends of the real-world population distribution witnessed over the past few centuries \citep{Tabuchi-RSUE2014}. 
However, a further decline in inter-regional transportation cost induces the flattening of a mono-centric agglomeration \citep{Helpman-Book1998,Tabuchi-JUE1998} due to the rise of the relative importance of urban costs (e.g., higher land rent and commuting costs). Other dispersive forces are found e.g. in \cite{Fujita-Krugman-Venables-Book1999}, through the inclusion of transport costs in the traditional sector, or in \cite{Murata-JUE2003} with the addition of heterogeneous preferences regarding residential location, and also in models of input-output linkages such as \cite{Krugman-Venables-QJE1995} and \cite{venables1996equilibrium} where, at a certain point, high wages in the more industrialized region forces firms to relocate to the periphery.\footnote{See \cite{Fujita-Thisse-Book2013} for a more detailed description of the mechanisms in input-output linkages models.} Whatever the additional dispersion forces, they counteract the net agglomeration forces from the manufacturing sector. Since the latter tend to vanish for a sufficiently low level of transportation costs, industry will tend to re-disperse after an initial phase of agglomeration. This lends support to the hypothesis of a “spatial Kuznets curve” whereby market forces initially increase, and then decrease,
spatial inequalities \citep{gaspar2018prospective, gaspar2020history}.

There is ample evidence on the decline of peak population or the production level of cities when transportation access improves \citep[e.g.,][]{Baum-Snow-QJE2007,Baum-Snow-et-al-REStat2017}. 
We may thus infer that developed economies face a final stage, in which once established economic clusters dissolve. 
However, these theories, do not address how production externalities operate between locations because they deliberately focus on trade linkages as \textit{the} mode of inter-location interaction to investigate the role of pecuniary externalities, as well as to ensure tractability \citep{Fujita-Mori-PRS2005}. 
Given the importance of agglomeration economies in an increasing number of information-intensive economies, we propose a tractable theory that integrates production externalities into general equilibrium economic geography models. 

For simplicity and tractability, we assume that the externality matrix $\VtG$ is exogenously given. 
There are various possible micro-foundations for $\VtG$. 
For instance, at the urban to regional spatial scale, it may represent the roles played by \textit{passenger travel} that supports face-to-face contact.  
It may also be some aggregate measure, embedded in regional space, of an inter-\textit{individual} social network that supports information exchange and diffusion between regions. 
It can also be a reduced form of the decisions of \textit{big players}, such as large companies that open up branches in provincial cities or airline companies connecting major regional cities. 
All of the above are described by sophisticated models, so that the structure of $\VtG$ may be endogenously determined by micro-economic mechanisms. 
However, in our study, $\VtG$ is exogenous, so as to provide first-order insights into the workings of an additional inter-regional linkage other than goods trade. Although this is a simplifying assumption, we emphasize that it is enough to provide valuable insights  and  convey our main messages  on the role of production externalities in shaping the space economy in an intelligible  and parsimonious way.
This strategy is akin to that in the network game literature \citep{jackson2010social} that focuses on the role of the structure of the inter-individual social network or to that in the economic geography literature where it is a standard approach to assume an exogenous inter-regional proximity structure. 

Regarding how inter-regional productivity spillovers can arise, 
although various micro-foundations can be considered, we highlight the role of knowledge creation due to the interactions between \textit{different cultures}. 
According to \cite{Fujita-RSUE2007}, geography is an essential feature of knowledge creation and diffusion. 
For instance, people residing in the same region interact more frequently and thus contribute to developing the same regional set of cultural ideas. 
Since geographically distant regions tend to develop different cultures, the economy evolves according to the synergy from the \textit{interactions across different regions} (i.e., different cultures). 
That is, as emphasized by \cite{Duranton-Puga-AER2001}, knowledge creation and location are inter-dependent. 
\cite{Berliant-Fujita-RSUE2012} developed a model of spatial knowledge interactions and showed that higher cultural diversity and costly communication promote the productivity of knowledge creation, which corroborates the empirical findings of \cite{Ottaviano-Peri-JEG2006,Ottaviano-Peri-DP2008}, as well as the theoretical model of \cite{Ottaviano-Prarolo-JRS2009}. 
If interaction between different regions with different cultures promotes knowledge creation, a region with good access to \textit{passenger transport} will be more innovative and productive. 
Our flexible model integrates such effects into a general equilibrium framework with costly trade.

The proposed model can also be seen as a reduced form of endogenous transportation cost models \citep[e.g.,][]{Behrens-etal-JUE2009,Behrens-Picard-JIE2011,Jonkeren-etal-JEG2011}. 
This literature stream considers settings in which transportation costs can fall with the population concentration because of, for example, scale and/or density economies in transportation. 
In our framework, a region with higher social proximity to the other regions becomes  more productive, thereby reducing the market price therein. 
This effect can be considered as endogenous transportation costs. 

Technically, we build on the general analytical method for an economic geography model developed by \cite{ikeda2012spatial} and \cite{Akamatsu-Takayama-Ikeda-JEDC2012}, which has been recently synthesized in \cite{AMOT-DP2019}. 
Additionally, our four-region analysis is inspired by \cite{matsuyama2017geographical}, who considers various tractable geographical settings to investigate how the underlying geographical structure impacts the home-market effect in a multi-region economy. 
Our study is also related to \cite{Barbero-Zofio-NETS2016}, as they focus on the role played by the spatial topology of the underlying transportation network.

\section{The Model\label{sec:model}}

Consider an economy comprised of $n$ regions, and let $\ClN \Is \{1,2,\hdots, n\}$ denote the set of regions. 
The economy is inhabited by a unit mass continuum of workers, who are freely mobile across regions. 
Each worker is inter-regionally mobile and may choose to reside in any of the $n$ regions.  
The spatial distribution of workers is denoted by $\Vtx = (x_i)_{i\inN}$, where $x_i\ge0$ is the mass of workers in region $i\inN$. 
The set of all possible $\Vtx$ is $\ClX \Is \{ \Vtx \ge \Vt0 \mid \sum_{i\inN} x_i = 1\}$.   

For simplicity, we assume that each region produces a distinct variety of horizontally differentiated goods, as in \cite{Armington-IMF1969}. 
Workers derive utility from the consumption of differentiated varieties. 
The workers are homogeneous and have identical constant-elasticity-of-substitution (CES) preferences over the differentiated varieties. 
As the utility function is homothetic, 
the total welfare in region $j\inN$ is  
\begin{align}
	U_j = 
	\left(
	    \sum_{i\inN}
	    q_{ij}^\frac{\sigma - 1}{\sigma}
    \right)^{\frac{\sigma}{\sigma - 1}}, 
\end{align}
where $\sigma > 1$ is the elasticity of substitution between varieties and $q_{ij}$ the amount of goods produced in region $i$ and consumed in region $j$. 

Production is perfectly competitive and 
labor is the only input factor. 
Each worker inelastically provides a unit of labor in the region they live and is compensated with a wage. 
The nominal market wage in region $j\inN$ is denoted by $w_j \ge 0$, and its spatial pattern by $\Vtw = (w_i)_{i\inN}$. 
The wage is determined at market equilibrium, which we describe later. 

\begin{remark}
It is known that the Armington-based economic geography framework is mathematically isomorphic to monopolistically competitive economic geography models such as \cite{Krugman-JPE1991} or \cite{Helpman-Book1998} \citep{Allen-Arkolakis-QJE2014,AMOT-DP2019}. 
Therefore, considering imperfect competition does not alter our results on stable spatial configurations. 
\end{remark}

The only agglomeration force in the model comes from production externalities. 
We assume the productivity of region $i$ is given as:
\begin{align}
    a_i(\Vtx) = 
    \sum_{j\inN} \psi_{ij} x_j,
    \label{eq:a_i} 
\end{align}
where $\psi_{ij}$ is the productivity spillover level from $j$ to $i$. 
Let $\VtG \Is [\psi_{ij}]$ be the \textit{externality matrix}, so that $\Vta(\Vtx) = (a_i(\Vtx))_{i\inN} = \VtG\Vtx$. 
This represents the (weighted) \emph{network structure} of inter-regional productivity spillovers. 

We introduce some assumptions on $\VtG$ as follows. 
\begin{assumption}
\label{assum:a_i}
Externality matrix $\VtG = [\psi_{ij}]$ satisfies the following property: 
\begin{enumerate}
    \item $\psi_{ij} \in(0,1]$ for all $i,j\inN$ with $\psi_{ii} = 1$, and 
    \label{enum:positivity}
    \item $\Vtz^\top\VtG\Vtz = \sum_{ij}\psi_{ij} z_i z_j > 0$ for any $\Vtz = (z_i)_{i\inN}$ such that $\sum_{i\inN} z_i = 0$. 
    \label{enum:CPD}
\end{enumerate}
\end{assumption}
\prettyref{assum:a_i}~\prettyref{enum:positivity} requires $\psi_{ij} > 0$ for all $i,j\inN$ but is not restrictive because $\psi_{ij}$ can be arbitrarily close to zero. 
Under \prettyref{assum:a_i}~\prettyref{enum:CPD}, $\Vta(\cdot)$ exhibits positive effects of agglomeration \citep[see, e.g.,][]{Osawa-Akamatsu-2020}. 
For example, consider an infinitesimal relocation of workers from $j$ to $i$, represented by a vector $\Vtz = \epsilon(\Vte_i - \Vte_j)$ with $\epsilon > 0$ and $\Vte_i$ being $i$th standard basis. 
Under \prettyref{assum:a_i}~\prettyref{enum:CPD}, the gain in $a_i$ induced by infinitesimal migration $\Vtz$ is strictly greater than the gain in $a_j$ because $a_i(\Vtx + \Vtz) - a_i(\Vtx) > a_j(\Vtx + \Vtz) - a_j(\Vtx)$. 
Therefore, the relocation of workers has \textit{self-reinforcing effects} in terms of regional productivity.

The inter-regional transportation of goods is costly. 
We assume iceberg transportation costs, that is, $\tau_{ij} \ge 1$ units should be shipped from $i$ for a unit to arrive at $j$, with $\tau_{ii} = 1$. 
Under perfect competition, the price of the good produced in $i$ and consumed in $j$ is 
\begin{align}
    p_{ij} = \frac{w_i}{a_i}\tau_{ij}. 
\end{align}
A higher freeness of production externalities increases productivity, and firms with higher worker productivity face lower marginal costs and thus charge a smaller optimal price, $p_{ij}$.
Under CES, the \textit{value} shipped from location $i$ to $j$ is given by
\begin{align}
	Q_{ij}
	= \frac{p_{ij}^{1 -\sigma}}{P_j^{1 - \sigma}} w_j x_j,
\end{align}
where $P_j$ is the CES price index: 
\begin{align}
	P_j 
	= 
	\left( 
    	\sum_{i\inN}
        p_{ij}^{1 - \sigma}
	\right)^{\frac{1}{1 - \sigma}}
	= 
	\left( 
	    \sum_{i\inN}
		    a_i^{\sigma - 1} w_i^{1 - \sigma} \phi_{ij}
	\right)^{\frac{1}{1 - \sigma}}, 
\end{align}
with $\phi_{ij} \Is \tau_{ij}^{1 - \sigma} \in (0,1]$. 
We denote $\VtD = [\phi_{ij}]$ and call $\VtD$ the \textit{geographical proximity matrix}. 
We assume all entries of $\VtD$ are strictly positive, that is, $\tau_{ij} < \infty$. 
\begin{assumption}
$\phi_{ij}\in (0,1)$ for all $i,j\inN$ if $i\ne j$ and $\phi_{ii} = 1$ for all $i\inN$.
\label{assum:phi}
\end{assumption}
\noindent We consider Assumptions \ref{assum:a_i} and \ref{assum:phi} to hold throughout the paper. 

The regional price index is decreasing in $a_i$, implying that a higher level of production externalities decreases the cost of living in region $i$. 
As a result, global demand in region $i$ increases in the spatial extent of externalities. 
Markets clear if the regional income is equal to the value of the goods sold in all regions, that is, for all $i\inN$:%
\begin{align}
		w_i x_i 
			= 
			\sum_{k\inN} Q_{ik}
			= 
			\sum_{k\inN}
				\dfrac
				{a_i^{\sigma - 1} w_i^{1 - \sigma} \phi_{ik}}
				{\sum_{l\inN} a_l^{\sigma - 1} w_l^{1 - \sigma} \phi_{lk}}
			w_k x_k. 
		\label{eq:short-run-equilibrium} 
\end{align}
\begin{remark}
If we let $\Htphi_{ij}(\Vtx) \Is a_i(\Vtx)^{\sigma - 1}\tau_{ij}^{1 - \sigma}$, it can be considered a reduced form of the trade cost that depends on the population distribution.%
\footnote{We thank Kristian Behrens for suggesting this natural and intuitive interpretation.}
Hence, our model can also be interpreted as a model with endogenous transportation costs, in which transportation costs, which fall with the concentration of population, possibly because of scale and/or density economies in transportation.  
\end{remark}
To normalize $\Vtw$, we assume that the total income of the economy is unity:
\begin{align}
    \sum_{i\inN} w_i x_i = 1.    
    \label{eq:normalization}
\end{align} 
Assume $\Vtx$ is positive, that is, $x_i > 0$ for all $i\inN$. 
Then, there is a unique wage vector $\Vtw$ that solves the market equilibrium conditions  \prettyref{eq:short-run-equilibrium} and \prettyref{eq:normalization} (see \prettyref{app:proofs}).

Market wage $\Vtw$ is thus defined by \prettyref{eq:short-run-equilibrium} and \prettyref{eq:normalization} as an implicit function of $\Vtx$. 
Additionally, the wage in a region diverges as the mass of workers goes to zero. 

With market wage $\Vtw(\Vtx)$,the per capita indirect utility in region $i$ is given by
\begin{align}
    v_i(\Vtx) = \frac{w_i(\Vtx)}{P_i(\Vtx)}, 
\end{align}
where price index $P_i$ is also a function of $\Vtx$ through $\Vta(\Vtx)$ and $\Vtw(\Vtx)$. 
We use $\Vtv(\Vtx) = (v_i(\Vtx))_{i\inN}$ to denote the indirect utility or payoff of workers as the function of spatial distribution $\Vtx$. 
It can be shown by the implicit function theorem that there exists $\Vtw$ differentiable in $\Vtx$ whenever $x_i > 0$ for all $i\inN$.

The exogenous model parameters are the elasticity of substitution $\sigma>1$, the geographical proximity matrix $\VtD$, and the externality matrix $\VtG$. 
Given parameter values, the \textit{spatial equilibrium} of the model is a spatial distribution $\Vtx$ of workers, and its associated market wage $\Vtw$ that satisfies \prettyref{eq:short-run-equilibrium}, equalizing the utility of mobile workers across regions. In other words, the spatial distribution is a spatial equilibrium if no mobile worker in region $i$ has an incentive to move to another region, $j\neq i$.
The following result establishes the existence of a spatial equilibrium. 
\begin{proposition}
\label{prop:existence}
There exists a spatial equilibrium for any $\sigma > 1$. 
All spatial equilibria are positive, that is, all regions are populated at any spatial equilibrium. 
\end{proposition}
Because there is positive demand for all varieties at any finite level of transportation cost under the Armington assumption and labor is the only input, the market wage of any individual worker diverges when the mass of workers in the same region goes to zero and the worker's utility goes to infinity. 
That is, $v_i(\Vtx)$ satisfies Inada's condition in $x_i$ as $\lim_{x_i\to 0} v_i(\Vtx) = \infty$. 
Therefore, no spatial equilibrium can incorporate depopulated regions.

We show general properties for the extreme values of transportation costs. 
Obviously, a uniform distribution of workers across regions is an equilibrium if trade is completely frictionless and the level of production externalities is the same across locations. 
We formalize this as follows. 
\begin{proposition}
\label{prop:dispersion}
Consider the ``death-of-distance'' limit, where trade and the interaction between different regions are completely costless, that is, the limit when $\phi_{ij} \to 1$ and $\psi_{ij} \to 1$ for all $i,j\inN$. 
Then, uniform distribution $\BrVtx = (\Brx,\Brx,\hdots,\Brx)$  with $\Brx \Is \frac{1}{n}$ is the unique and stable spatial equilibrium. 
\end{proposition}
For all stability claims, we assume a class of myopic dynamics consistent with model (see \prettyref{app:stability}).

In the converse limit, where trade is too costly and production externalities are local within each region, the only stable equilibrium is full agglomeration in either of the regions. 
\begin{proposition}
\label{prop:mono-centric}
Consider the ``autarky'' limit, where trade between different regions is prohibitively costly and there is no production externalities between them (i.e., the limit when $\phi_{ij} \to 0$ and $\psi_{ij} \to 0$ for all $i \ne j$). 
Stable equilibrium spatial patterns are full agglomeration in one region, that is, $x_i = 1$ for some $i\inN$ and $x_j = 0$ for all $j \ne i$. 
\end{proposition}
Propositions \ref{prop:dispersion} and \ref{prop:mono-centric} show that the economy describes the dispersion process from a mono-centric configuration to uniform dispersion when transportation costs decline. 
We confirm this through examples in Sections \ref{sec:two-region_economy} and \ref{sec:four-region_economies}. 
\prettyref{prop:mono-centric} demonstrates that multiple equilibria can exist, as expected.

Below, we explore concrete examples to illustrate the essential implications of considering the additional proximity structure (i.e., $\VtG$). 
\prettyref{sec:two-region_economy} considers the canonical starting point, the symmetric two regions. 
\prettyref{sec:four-region_economies} considers four-region setups with several symmetric, yet representative, structures for externality matrix $\VtG$.

\section{Dispersion process in a two-region economy\label{sec:two-region_economy}}

We first consider a symmetric two-region economy, where the two regions have the same characteristics, to elucidate the basic workings of the model when transportation costs decline. 
When transport is prohibitive, the economy starts from an almost fully agglomerated state in either region (\prettyref{prop:mono-centric}). 
Uniform distribution $\BrVtx$ is stable when transport is free (\prettyref{prop:dispersion}); 
thus, $\BrVtx$ becomes stable at some intermediate transportation cost level. 
Here, we study the stability of the symmetric equilibrium and derive the critical level of transportation costs at which the equilibrium interchanges stability.

When regions are symmetric, it is natural to assume that  
\begin{align}
	\VtD 
	=
	\begin{bmatrix}
		1 & \phi \\ \phi & 1 
	\end{bmatrix}
	\quad
	\text{and}
	\quad
	\VtG
	=
	\begin{bmatrix}
		1 & \psi \\ \psi & 1 
	\end{bmatrix}, 
	\label{eq:D-G_2}
\end{align}
where $\phi \in (0,1)$ and $\psi\in(0,1)$. 
We call $\phi$ the \textit{freeness of trade} and $\psi$ the \textit{spatial extent of externalities}. 
When $\psi$ is small, it means that productivity spillovers attenuate sharply. 
It is obvious that uniform distribution $\BrVtx \Is (\Brx,\Brx)$ with $\Brx = \frac{1}{2}$ is a spatial equilibrium for all $(\sigma,\phi,\psi)$, where we have $w_1(\BrVtx) = w_2(\BrVtx) = 1$ and $v_1(\BrVtx) = v_2(\BrVtx) = \Brv \Is (1 + \psi)(1 + \phi)^{\frac{1}{\sigma - 1}}\Brx > 0$. 

\subsection{Net agglomeration forces at the uniform distribution}

In a two-region economy, migration from one region to the other is represented by the vector $\Vtz \Is (1,-1)$. 
For a wide range of dynamics, $\BrVtx$ is linearly stable if eigenvalue of $\nabla\Vtv(\BrVtx)$, the Jacobian matrix of $\Vtv$ at $\Vtx = \BrVtx$, associated with $\Vtz$, is negative; 
alternatively, it is linearly unstable if the eigenvalue is positive.\footnote{If the eigenvalue is zero, the equilibrium is non-hyperbolic, or otherwise called \emph{irregular}. According to \cite{Castro_2021}, non-existence of irregular equilibria is generic, that is, it holds in a full measure subset of a suitably defined parameter space.}
Below, for slight economy of notation, we consider the eigenvalue of the \emph{payoff elasticity} matrix with respect to $\Vtx$, $\VtV \Is \frac{\Brx}{\Brv}\nabla\Vtv(\BrVtx)$, which is a positive scalar multiple of the eigenvalue of $\nabla\Vtv(\BrVtx)$. 

Let $\omega$ be the eigenvalue of $\VtV$ associated with $\Vtz$. 
We can interpret $\omega$ as the \textit{net agglomeration force} at $\BrVtx$. 
Specifically, $\omega$ is the elasticity of the \emph{payoff difference} $v_1(\Vtx) - v_2(\Vtx)$ with respect to $x_1$ at $\BrVtx$, as we can write $\omega$ as follows: 
\begin{align}
	\label{eq:omega_2reg}
	\omega 
	& 
	=
	\frac{\Brx}{\Brv}
	\PDF{\left(v_1 - v_2\right)}{x_1}(\BrVtx)
	= 
	\frac{\Brx}{\Brv}
	\left(
	\PDF{v_1}{x_1}(\BrVtx)
	-
	\PDF{v_2}{x_1}(\BrVtx)
	\right). 
\end{align}
If $\omega < 0$, there is no incentive for agents to migrate because a marginal increase in the mass of workers in a region induces a \emph{relative decrease} of the utility therein. 
Similarly, $\BrVtx$ is unstable if $\omega > 0$; 
when a small fraction of workers relocate from region $2$ to $1$, it induces a \emph{relative increase} of the payoff in region $1$, thus encouraging further migration from region $2$. 

We can break down $\omega$ by the chain rule as 
\begin{align}
    \omega 
    = 
        \omega_a
        \alpha_x
        +
        \omega_w
        \beta_x, 
\end{align}
where $\omega_a$ and $\omega_w$ are the elasticities of payoff difference $v_1(\Vtx) - v_2(\Vtx)$ with respect to $a_1$ and $w_1$ at $\Vtx = \BrVtx$, and 
$\alpha_x$ and $\beta_x$ are the elasticities of $a_1$ and $w_1$ with respect to migration of workers from one region to the other, respectively.%
\footnote{Specifically, 
with $\Bra \Is a_1(\BrVtx) = (1 + \psi)\Brx$ and $\Brw = 1$ being the uniform level of regional productivity and wage, 
\begin{align*}
	& 
    \omega_a
    = 
    \frac{\Bra}{\Brv}
    \left(
    \PDF{v_1}{a_1}(\BrVtx)
    -
    \PDF{v_2}{a_1}(\BrVtx)
    \right), 
    \omega_w
    = 
    \frac{\Brw}{\Brv}
    \left(
    \PDF{v_1}{w_1}(\BrVtx)
    -
    \PDF{v_2}{w_1}(\BrVtx)
    \right),
    \text{ and }
    \alpha_x
    =  
    \frac{\Brx}{\Bra}
    \left(
        \PDF{a_1}{x_1}(\BrVtx)
        - 
        \PDF{a_1}{x_2}(\BrVtx)
    \right). 
\end{align*}
We can swap the regional indices in the above expressions due to regional symmetry.
} 

For $\omega_a$ and $\omega_w$, we can show that 
$\omega_a = \chi$ and $\omega_w = 1 - \chi$, 
where $\chi$ is given by
\begin{align}
    \chi \Is \frac{1 - \phi}{1 + \phi} \in(0,1). 
\end{align} 
We observe that $\chi$ is an index of trade costs, which is close to $1$ when $\phi$ is small and goes to $0$ when $\phi$ approaches $1$. 
Both $\omega_a$ and $\omega_w$ are positive and respectively indicate positive effects of regional productivity and wage on the regional payoff. 
When trade is more costly ($\phi$ is small $\Leftrightarrow$ $\chi$ is large), the payoff difference is more sensitive to the variation in regional productivity ($\omega_a = \chi$ is large), whereas it is less sensitive to wage ($\omega_w = 1 - \chi$ is small). 

We have 
$\alpha_x = \lambda$
where $\lambda$ is given by 
\begin{align}
    \lambda \Is \frac{1 - \psi}{1 + \psi} \in(0,1), 
\end{align}
which is the \emph{sensitivity} of the inter-regional production externality level to an infinitesimal migration of workers from one region to the other. 
A region's productivity is sensitive to migration when inter-regional production spillover is weak because $\alpha_x = \lambda$ is large when $\psi$ is small; the opposite also holds true.

For $\beta_x$, we compute it as follows: 
\begin{align}
    & 
    \beta_x 
    = 
    \frac{1}{\sigma + (\sigma - 1)\chi}
    \left(
    (\sigma - 1)(1 + \chi)\alpha_x - 1 
    \right). 
    \label{eq:beta_x}
\end{align}
A population increase in region $1$ induces a relative productivity increase $\alpha_x = \lambda > 0$, which in turn increases the nominal wage in the region, since $\omega_a \alpha_x > 0$. 
However, if the mass of workers in a region increases, then it has a negative effect on the nominal wage, since the total revenue in a region is given by $w_i x_i$. 

Among these elasticities, only $\beta_x$ can be negatively associated with regional asymmetry and produce stabilizing effects. 
If $1 > (\sigma - 1)(1 + \chi)\alpha_x = (\sigma - 1)(1 + \chi)\lambda$, then $\beta_x$ is negative; 
this can happen when $\chi$ and $\lambda$ are sufficiently small, that is, when $\psi$ and $\phi$ are sufficiently large (and/or $\sigma$ is sufficiently small). 

In sum, $\omega$ is represented as $\omega = \Omega(\chi,\lambda)$, where $\Omega$ is defined follows: 
\begin{align}
    \Omega(s,t)
    \Is 
    \frac{
    -
    (1 - s) 
    + 
    \left((\sigma - 1) + \sigma s \right)
    t
    }{\sigma + (\sigma - 1) s} 
    \label{eq:Omega}. 
\end{align}
The denominator of $\Omega$ comes from $\beta_x$ in \prettyref{eq:beta_x} and is positive for all admissible values of $\sigma$ and $\phi$. 

The numerator of $\omega = \Omega(\chi,\lambda)$ reveals the net agglomeration and dispersion forces in the model. 
The first term, $-(1 - \chi)$, is negative and thus represents the dispersion force due to costly trade. 
The dispersion force \textit{strengthens} when $\phi$ increases, since $\chi$ is decreasing in $\phi$. 
Because every region specializes in a single variety, workers' love for variety due to CES preferences induces a stronger centrifugal force when trade is freer. 
The second term, $\left((\sigma - 1) + \sigma\chi \right)\lambda$, is positive and represents the agglomerative force due to productivity spillover \prettyref{eq:a_i}. 
Because $\lambda\in(0,1)$ is monotonically decreasing in $\psi$, this force is at its strongest when $\psi$ is small, which is intuitive. 
Overall, $\omega$ tends to be positive when $\phi$ and $\psi$ are smaller (i.e., when transport is costly), which is consistent with \prettyref{prop:mono-centric}. 

\begin{remark}
We can remove agglomerative effects given expression \prettyref{eq:Omega} by setting $t = 0$. 
If regional productivity $\Vta$ is a $\Vtx$-independent constant, then $\alpha_x = 0$. 
For this case, $\BrVtx$ is always stable because 
\begin{align}
	\omega = \omega_w \beta_x = 
	- \dfrac{1 - \chi}{\sigma + (\sigma - 1)\chi} < 0. 
\end{align}
Without any agglomerative forces, costly trade discourages the uneven concentration of workers. 
\end{remark}

\begin{remark}
If we consider reduced-form congestion effects for local amenities as in \cite{Allen-Arkolakis-QJE2014} by a payoff function such as $\TlU_j = x_j^{-\theta}U_j$ with $\theta > 0$, then the relevant eigenvalue becomes $\Tlomega = - \theta + \omega$. 
For simplicity, we do not consider such congestion effects. 
\end{remark}

\begin{figure}[tb]
	    \centering
	    \includegraphics{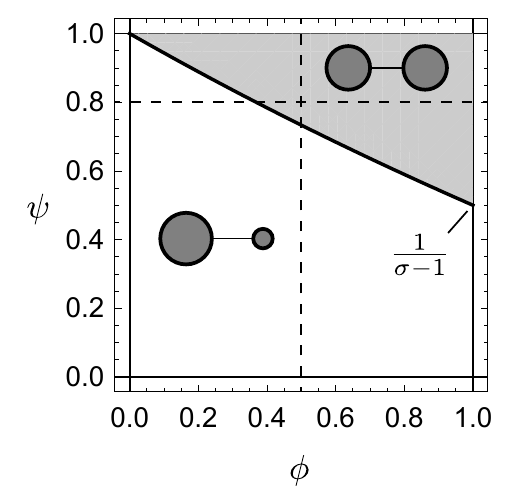}
	\caption{Stability of $\BrVtx$ in a two-region economy ($\sigma = 4.0$).}
	\label{fig:stab-unif-2}
	\begin{figurenotes}
	Uniform distribution $\BrVtx$ is stable for the shaded (gray) region of $(\phi,\psi)$ and the black solid curve indicates the critical pair of $(\phi,\psi)$, where $\BrVtx$ becomes unstable. 
	The horizontal and vertical dashed lines correspond to the parametric paths for bifurcation diagrams \prettyref{fig:bif-2-phi} and \prettyref{fig:bif-2-psi}, respectively. 
	The schematic on each (gray or white) parametric region indicates the representative spatial pattern in that parametric region. The results remain invariant under a reasonable range of values for sigma (see \prettyref{app:sigma-variation}).
	\end{figurenotes}
\end{figure}

\subsection{Stability of dispersion}

From the formula of $\omega$, we have the following characterization for the stability of $\BrVtx$. 
\begin{proposition}
\label{prop:n=2}
Assume $n = 2$ and $\VtD$ and $\VtG$ in \prettyref{eq:D-G_2}. Then, uniform distribution $\BrVtx = (\Brx,\Brx)$ is linearly stable if and only if $\omega = \Omega(\chi(\phi),\lambda(\psi)) 
    < 0$.
\end{proposition}

\prettyref{fig:stab-unif-2} shows \prettyref{prop:n=2} on the $(\phi,\psi)$-space. 
Uniform distribution $\BrVtx$ is stable in the shaded areas of $(\phi,\psi)$, and unstable otherwise. 
The black solid curve shows the critical pairs of $(\phi,\psi)$ below which $\BrVtx$ becomes unstable, that is, the solutions for  $\omega(\phi,\psi) = \Omega(\chi(\phi),\lambda(\psi)) = 0$. 
Stability condition $\omega < 0$ is satisfied when both $\phi$ and $\psi$ are relatively high. 
For any $\phi\in(0,1)$, $\BrVtx$ is stable when the spatial extent of externalities $\psi$ is sufficiently high.

\begin{figure}[tb]
	\centering
	\begin{subfigure}[c]{.32\hsize}
	    \centering
	    \includegraphics[width=\hsize]{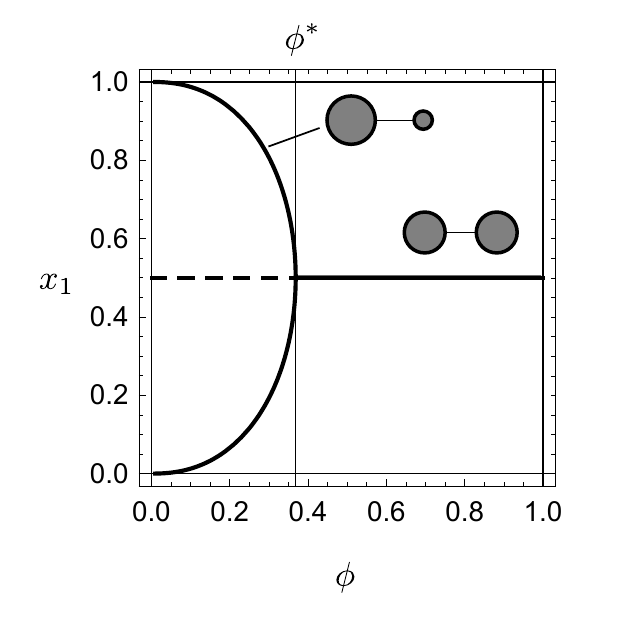}
	    \caption{$(\psi,\sigma) = (0.8,4.0)$} 
	    \label{fig:bif-2-phi}
	\end{subfigure}
	\begin{subfigure}[c]{.32\hsize}
	    \centering
	    \includegraphics[width=\hsize]{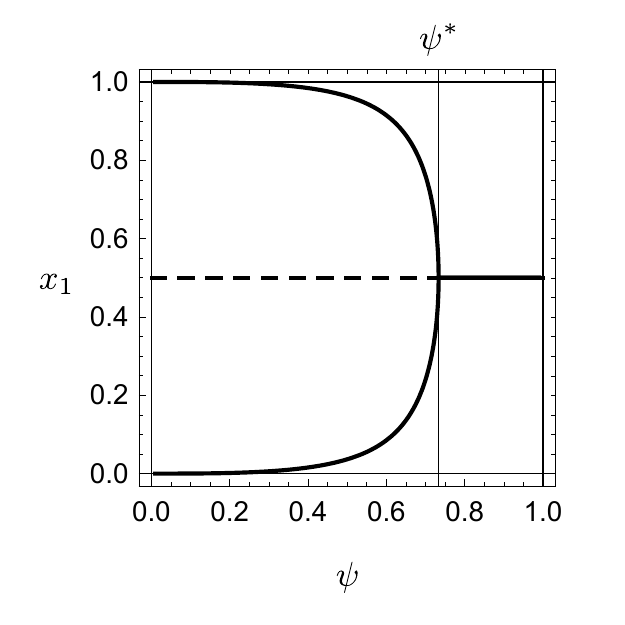}
	    \caption{$(\phi,\sigma) = (0.5,4.0)$} 
	    \label{fig:bif-2-psi}
	\end{subfigure}
	\begin{subfigure}[c]{.32\hsize}
	    \centering
	    \includegraphics[width=\hsize]{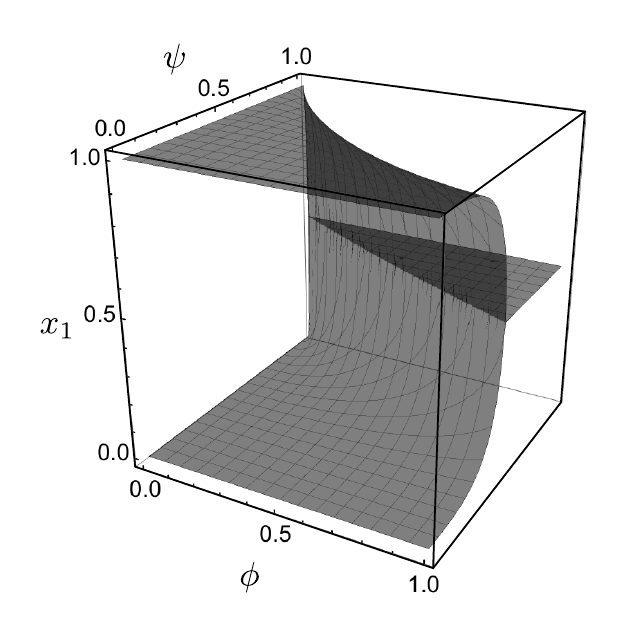}
	    \caption{Full picture ($\sigma = 4.0$)} 
	    \label{fig:bif-2-full}
	\end{subfigure}
	\caption{Bifurcation diagrams for a symmetric two-region economy.}
	\begin{figurenotes}
	In Figures \ref{fig:bif-2-phi} and \ref{fig:bif-2-psi}, 
	the black solid curves indicate stable spatial equilibria and the dashed curves indicate unstable ones. 
	The corresponding spatial configurations are schematically shown, where the size of each gray disk indicates the size of a region. 
	In \prettyref{fig:bif-2-full}, the transparent gray surface indicates stable equilibria in terms of $x_1$. 
	Figures \ref{fig:bif-2-phi} and \ref{fig:bif-2-psi} are the cross sections of the surface at $\psi = 0.8$ and $\phi = 0.5$, respectively. 
	We can consider various cross sections of the surface, or curves over the $(\phi,\psi)$-space to investigate the effects of simultaneous changes in $(\phi,\psi)$. 
	\end{figurenotes}
	\label{fig:bifs-2}
\end{figure}

The boundary of the gray region in \prettyref{fig:stab-unif-2} is represented as follows:
\begin{align}
    \phi^* = (2\sigma - 1) \frac{\lambda}{2 + \lambda} = (2\sigma - 1) \frac{1 - \psi}{3 + \psi} > 0. 
\end{align} 
If $\phi^* \in (0,1)$, then the uniform distribution is stable for all $\phi \in (\phi^*,1)$. 
If, otherwise, $\phi^* \ge 1$, then $\BrVtx$ is unstable for all $\phi\in(0,1)$, so that the economy always exhibits asymmetry (e.g., $x_1 > x_2$). 
We require $\phi^* \in (0,1)$, which may be called the ``no-black-hole'' condition, following \cite{Fujita-Krugman-Venables-Book1999}. 
We have $\phi^*\in(0,1)$ either when $\sigma\in(1,2)$ (which is unrealistic) or when $\sigma \ge 2$ and $\psi > \frac{\sigma - 2}{\sigma}$. 
\noindent If $\psi \le \frac{\sigma  - 2}{\sigma}$, then $\BrVtx$ is unstable for any $\phi$. 
When elasticity of substitution $\sigma$ is relatively large and freeness of interaction $\psi$ is relatively small, then migration toward one of the regions is profitable. 
Our model shows that, when $\phi$ gradually increases, $\BrVtx$ becomes stable when $\phi^*$ is attained; a similar conclusion can be derived for the $\psi$-axis.

\prettyref{fig:bifs-2} shows the bifurcation diagram of the stable spatial equilibria in terms of $x_1$ when $\phi$ and/or $\psi$ vary. 
Figures \ref{fig:bif-2-phi} and \ref{fig:bif-2-psi} respectively show the bifurcation diagrams for the horizontal and vertical dashed lines in \prettyref{fig:stab-unif-2}. 
\prettyref{fig:bif-2-full} shows the bifurcation diagram of the stable equilibrium values of $x_1$ over the full $(\phi,\psi)$-space. 
The uniform distribution is stable for high $\phi$ or $\psi$ values. 
The model represents the resolution process of an established agglomeration when $\phi$ and/or $\psi$ monotonically increase.  
The stable equilibrium paths are continuous on the $\phi$ and $\psi$ axes and there are no catastrophic jumps nor hysteresis when $\BrVtx$ becomes unstable. 
These properties are akin to \cite{Helpman-Book1998}'s model, with an additional dimension of $\psi$.
We can formally show that the bifurcation from $\BrVtx$ is a supercritical pitchfork, which essentially means that the dispersion process is ``reversible.'' 
\begin{proposition}
\label{prop:pitchifork}
The bifurcation from $\BrVtx$ when decreasing $\phi$, or decreasing $\psi$, takes a supercritical pitchfork form. 
The dispersion process of economic activities is smooth and gradual as the economy becomes more symmetric.  
\end{proposition}

\begin{figure}[tb]
    \centering 
	\begin{subfigure}[c]{.4\hsize}
	    \centering
	    \includegraphics[width=.8\hsize]{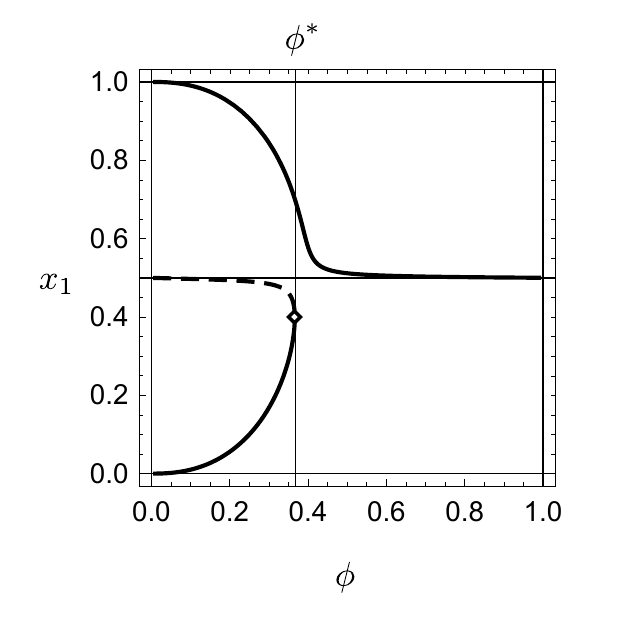}
	    \caption{$(\psi,\sigma) = (0.8,4.0)$, $\phi_{12} = \phi^{1.1} < \phi$} 
	    \label{fig:bifs-2asym_phi}
	\end{subfigure}
	\begin{subfigure}[c]{.4\hsize}
	    \centering
	    \includegraphics[width=.8\hsize]{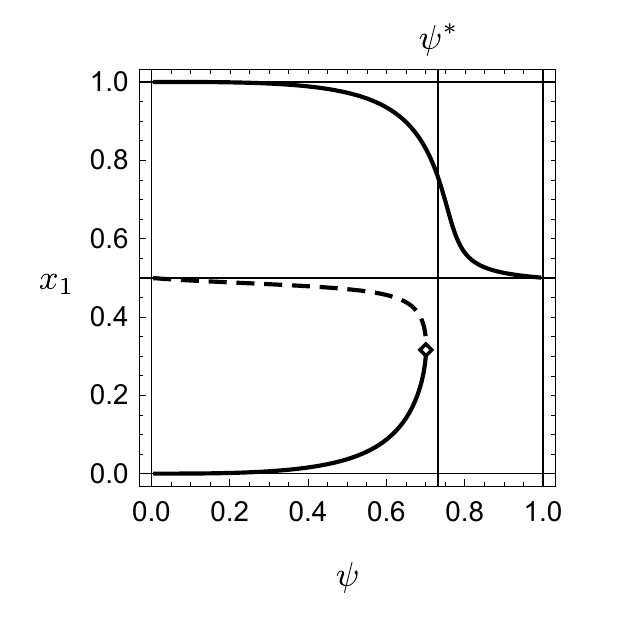}
	    \caption{$(\phi,\sigma) = (0.5,4.0)$, $\psi_{21} = \psi^{1.1} < \psi$} 
	    \label{fig:bifs-2asym_psi}
	\end{subfigure}

	\caption{Bifurcation diagrams for an asymmetric two-region economy.}
	
	\begin{figurenotes}
	 The black solid curves indicate stable spatial equilibria and the dashed curves indicate unstable ones. The diamond ($\lozenge$) in each figure indicates the limit point from which a pair of unstable and stable equilibria emerges. $\phi^*$ and $\psi^*$ are the critical values for the symmetric cases in \prettyref{fig:bifs-2}.
	\end{figurenotes}
	\label{fig:bifs-2asym}
\end{figure}

Asymmetries in the proximity matrices ($\VtD$ and $\VtG$) induce straightforward comparative advantages. 
\prettyref{fig:bifs-2asym} shows examples under which the two regions are asymmetric. 
In Figures \ref{fig:bifs-2asym_phi} and \ref{fig:bifs-2asym_psi}, we respectively assume $\VtD$ and $\VtG$ are of the form 
\begin{align}
    \VtD 
    =
    \begin{bmatrix}
        1 & \phi^{1.1} \\
        \phi & 1
    \end{bmatrix}
    \text{\quad and\quad}
    \VtG 
    =
    \begin{bmatrix}
        1 & \psi \\
        \psi^{1.1} & 1
    \end{bmatrix}. 
\end{align}
For both cases, region $1$ has a comparative advantage (in terms of the cost of living or productivity). 
The bifurcation diagrams exhibit the standard \emph{unfolding} behavior for the supercritical pitchfork bifurcation, for which the transition on the main path (the path with $x_1 > x_2$) is smooth, without any catastrophic behaviors. 
The economy is always asymmetric, and the uniform distribution emerges only within limits $\phi \to 1$ or $\psi \to 1$. 
For regional models, \cite{Berliant-Kung-RSUE2009} provides a detailed analysis of such unfolding behavior for multiple heterogeneity parameters.

\section{How the network structure of externality matrix matters\label{sec:four-region_economies}}

The two-region setting in \prettyref{sec:two-region_economy} demonstrates how our model captures the basic dispersion process in a regional economy. 
We now turn our attention to a multi-region geography to identify the roles of the externality structure $\VtG$: (i) the \emph{timing} of dispersion and (ii) the overall \emph{spatial distribution} of workers. 

We consider a symmetric geography in which $n = 4$ regions are equidistantly placed over a circular transportation network (see \prettyref{fig:4-reg}), as in \cite{matsuyama2017geographical}, Example 2, which can be seen as a simplified version of the $12$-location race-track economy of \cite{Krugman-EER1993}. 
This is the minimal symmetric geographical environment in which different regions have different neighbors. 
By postulating that the transportation of goods is only possible over the circular network, we can assume that the geographical proximity matrix is given by
\begin{align}
	\VtD 
	=
	\begin{bmatrix}
		1 & \phi & \phi^2 & \phi\\
		 & 1 & \phi & \phi^2\\ 
		 &  & 1 & \phi\\
		\SYM &  & 1
	\end{bmatrix},
	\label{eq:D-racetrack}
\end{align}
where $\phi \in (0,1)$ is the freeness of trade between two neighboring regions in the economy. 
Under this setting, all regions have the same level of geographical proximity to the other regions.

\begin{figure}[tb]
	\centering
	\includegraphics{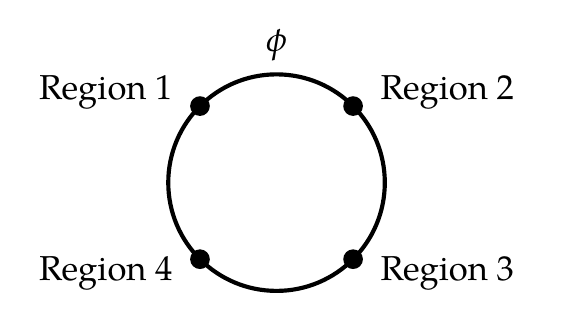}
	\caption{Four-region symmetric geography.}
	\begin{figurenotes}
	The black circle represents the transportation network, whereas the black markers represent the regions. All regions have the same geographical proximity level to the other regions. 
	\end{figurenotes}
	\label{fig:4-reg}
\end{figure}

\begin{figure}[tb]
	\centering
	\begin{subfigure}[b]{.32\hsize}
		\centering 
		\includegraphics{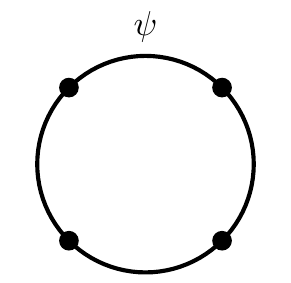}
		\caption{Simple}
		\label{fig:g1}
	\end{subfigure}
	\begin{subfigure}[b]{.32\hsize}
		\centering
		\includegraphics{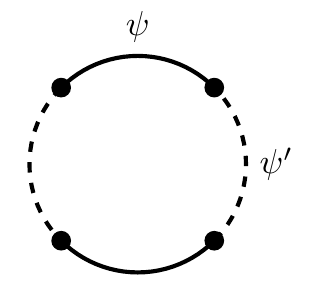}
		\caption{Block economy ($\psi' \le \psi$)}
		\label{fig:g3}
	\end{subfigure}
	\begin{subfigure}[b]{.32\hsize}
		\centering 
		\includegraphics{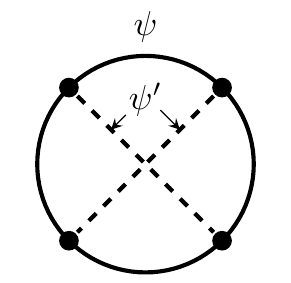}
		\caption{Bypass ($\psi' \ge \psi^2$)}
		\label{fig:g2}
	\end{subfigure}
	\caption{Symmetric externality structures in the four-region geography.}
	\begin{figurenotes}
	(a) The baseline case where production externalities are only governed by geographical proximity. 
	(b) A case where the externality matrix has a hierarchical structure.
	(c) A case where the pairs of regions at the antipodal locations on the circle are strongly tied. 
	\end{figurenotes}
	\label{fig:four-region-networks}
\end{figure}

Four is the minimal number of regions that allows nontrivial structures of inter-regional externalities between locations while preserving symmetry. 
Taking $\VtD$ in \prettyref{eq:D-racetrack} as given, we consider three stylized settings for $\VtG$ to identify the basic roles of network structure (\prettyref{fig:four-region-networks}). 
\prettyref{fig:g1} shows the baseline case, in which the magnitudes of production externalities are governed by geographical proximity. 
\prettyref{fig:g3} represents the case where $\VtG$ has a block structure and 
\prettyref{fig:g2} the case where the antipodal locations are strongly tied in $\VtG$. 
 
For all three cases, $\VtG$ does not exhibit any ex-ante comparative advantage of the regions because it preserves the symmetry of the four-region circular economy. 
Therefore, combined with geographical symmetry (\prettyref{fig:4-reg}), uniform distribution $\BrVtx = (\Brx,\Brx,\Brx,\Brx)$ with $\Brx \Is \frac{1}{4}$ is always a spatial equilibrium. 

\begin{remark}
A realistic setup is a star-shaped network, where one of the regions (e.g., the host of the capital city) is the hub for the production externalities in the economy. 
For example, we may assume $\psi_{1j} = \psi_{j1} = \psi$ and $\psi_{jk} = \psi_{kj} = \psi' < \psi$ ($k,j\ne 1$) so that region $1$ has an exogenous advantage due to the higher productivity $a_1 > a_j$ ($j \ne 1$). 
This leads to a straightforward consequence:  
the spatial pattern becomes mono-centric, with region $1$ as the central location. 
We instead focus on symmetric networks in \prettyref{fig:four-region-networks} to focus on endogenous forces. 
\end{remark}
\begin{remark}
\label{remark:equi}
The equidistant geographical network considered by, for example, \cite{Gaspar-et-al-ET2018,Gaspar-et-al-IJET2019,Gaspar-et-al-RSUE2021} and \cite{Aizawa-etal-IJBC2020}, in which all regions are mutually connected (i.e., $\phi_{ii} = 1$ and $\phi_{ij} = 1$ for all $i \ne j$), has the advantage of being  a very tractable spatial setting, without necessarily being less realistic from an empirical view point when compared to other stylized geometries. 
However, in this setting, agents' incentives are determined by a region itself and the rest of the economy. Geographical asymmetries are non-existent and a thus more regions just adds to higher market access variability. As a result, multi-centric patterns do not become stable in most cases \citep{Aizawa-etal-IJBC2020} and the spatial distribution tends to be mono-centric, in contrast to circular geography where stable polycentric patterns can emerge \citep{AMOT-DP2019}.  
\end{remark}

\subsection{Agglomeration in the multi-region economy}

Because of the many possible spatial configurations in a four-region economy, we consider the process of dispersion in a ``reverse-reproduced'' way. 
We start from the transportation costs levels at which $\BrVtx$ is stable (i.e., large $\phi$ or $\psi$) and consider the monotonic \emph{increase} of transportation costs (i.e., decrease of $\phi$ or $\psi$).\footnote{In  economic geography  it is customary to observe the qualitative behavior of the model as transportation costs decrease so as to capture the impact of increasing economic integration or globalization. However, there is also ample empirical evidence of rising costs at different geographical scales.  \cite{RePEc:cap:wpaper:012021} for instance, study the emergence of satellite cities along a narrow corridor as a consequence of an increase in transportation costs that deems the existence of a single monocentric city economically unsustainable.}
If $\BrVtx$ becomes unstable, then some nontrivial spatial configurations emerge, providing insights into stable spatial patterns at intermediate values of transportation costs. 

In our multi-region setting, the stability of $\BrVtx = (\Brx, \Brx, \Brx, \Brx)$ is governed by the \emph{largest eigenvalue} of the payoff elasticity matrix $\VtV = \frac{\Brx}{\Brv} \nabla\Vtv(\BrVtx)$.
To explain this, assume that $\BrVtx$ is perturbed to become $\Vtx' = \BrVtx + \Vtz$ with small $\Vtz = (z_i)_{i\inN}$, where we require $\sum_{i\inN} z_i = 0$ because the total population of the economy is constant. 
In other words, $\Vtz$ is a migration pattern. 
\begin{figure}[tb]
	\centering
	\begin{subfigure}[b]{.32\hsize}
		\centering 
		\includegraphics[width=3cm]{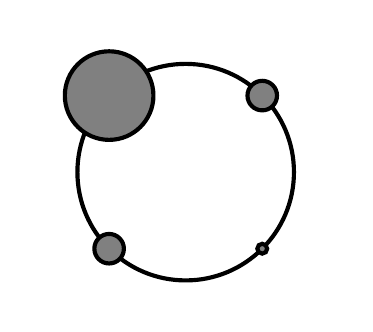}
		\caption{Mono-centric pattern}
		\label{fig:z1}
	\end{subfigure}
	\begin{subfigure}[b]{.32\hsize}
		\centering 
		\includegraphics[width=3cm]{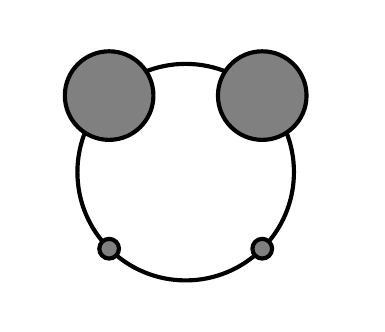}
		\caption{North--South pattern}
		\label{fig:z3}
	\end{subfigure}
	\begin{subfigure}[b]{.32\hsize}
		\centering
		\includegraphics[width=3cm]{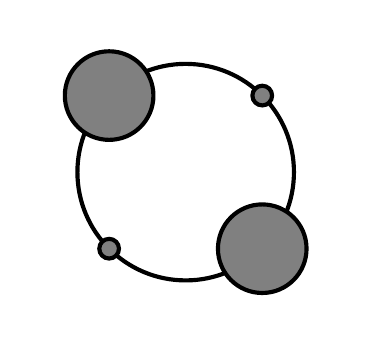}
		\caption{Duo-centric pattern}
		\label{fig:z2}
	\end{subfigure}
	\caption{Schematics of the possible endogenous outcomes in a symmetric four-region economy.}
	\begin{figurenotes}
	The black circle indicates the transportation network. 
	The gray disk represent the population size of each region. 
	We do not show rotationally symmetric patterns essentially equivalent to the three patterns given above (e.g., the ``East--West'' pattern).  
	\end{figurenotes}
	\label{fig:four-region-patterns}
\end{figure}
The average gain (in terms of relative payoff) induced by such a deviation may be evaluated by
\begin{align}
    \Bromega(\Vtz)
    \Is 
    \frac{\Vtz^\top \VtV \Vtz\ }{\|\Vtz\|^2}, 
    \label{eq:baromega}
\end{align}
which can be seen as the (normalized) elasticity of average payoff $\sum_{i\inN} v_i(\Vtx) x_i$, 
\begin{align}
    \frac{\Brx}{\Brv}
    \left(\sum_{i\inN} v_i(\Vtx') x_i' - \sum_{i\inN} v_i(\BrVtx) \Brx \right)
    \approx 
    \Vtz^\top \VtV \Vtz.  
\end{align}
If $\Bromega(\Vtz) < 0$ for \textit{any} migration pattern $\Vtz$, then any form of migration is strictly non-profitable for migrants and $\BrVtx$ is stable. 

Average gain $\Bromega$ is maximized by choosing $\Vtz = \Vtz^*$, where $\Vtz^*$ is the eigenvector of $\VtV$ associated with its \emph{largest eigenvalue}, $\omega^* \Is \max_{k}\{\omega_k\}$, 
where $\{\omega_k\}$ are the eigenvalues of $\VtV$. 
That is, we have 
\begin{align}
    \max_{\Vtz} \Bromega(\Vtz)
    = 
    \Bromega(\Vtz^*) = \omega^*, 
\end{align}
If $\omega^* < 0$, then $\Bromega < 0$ for any $\Vtz$. 
When $\omega^*$ switches from negative to positive, associated migration pattern $\Vtz^*$ becomes profitable for workers and the spatial pattern of form $\Vtx' = \BrVtx + \epsilon\Vtz^*$ ($\epsilon > 0$) emerges. 
We thus have to determine $\omega^* \Is \max_k\{\omega_k\}$ and its associated eigenvector to identify what the spatial pattern that will emerge from $\BrVtx$. 

We show that, when one starts from $\BrVtx$ in our setting, three qualitatively different spatial distributions of workers can emerge: 
\begin{itemize}
	\item A mono-centric distribution of the form $(\Brx, \Brx + \epsilon, \Brx, \Brx - \epsilon)$, in which there are one big, one small, and two medium-sized cities (\prettyref{fig:z1}). 
	\item A North--South distribution of form $(\Brx + \epsilon,\Brx + \epsilon, \Brx - \epsilon, \Brx - \epsilon)$, in which two contiguous regions attract the majority of workers from the rest of the economy (\prettyref{fig:z3}). 
	\item A duo-centric distribution of form $(\Brx + \epsilon,\Brx - \epsilon, \Brx + \epsilon, \Brx - \epsilon)$, in which two geographically remote regions are vying with each other (\prettyref{fig:z2}). 
\end{itemize}
These spatial distributions respectively correspond to $\Vtz^* = (0,1,0,-1)$, $\Vtz^* = (1,1,-1,-1)$, and $\Vtz^* = (1,-1,1,-1)$, all of which are eigenvectors of $\VtV$. 
By contrast, when $n = 2$, the only possible migration pattern is $\Vtz^* = (1, -1)$ and $\omega$ is the only relevant eigenvalue of $\VtV$. 

\subsection{The baseline case}
\label{sec:4-baseline}

This section considers the case in \prettyref{fig:g1}, where the externality matrix has a similar structure to the geographical proximity matrix $\VtD$. 
This setup is related to the geographically decaying spillovers considered in urban economics models \citep[e.g.,][]{Fujita-Ogawa-RSUE1982}. 
The externality matrix is given as follows:
\begin{align}
	\VtG
	=
	\begin{bmatrix}
		1 & \psi & \psi^2 & \psi\\
		 & 1 & \psi & \psi^2\\ 
		 & & 1 & \psi\\
		\SYM   & & 1
	\end{bmatrix},  
	\label{eq:G-baseline}
\end{align}
where $\psi \in (0,1)$. 
In this case, we have $\Vtz^* = (1, 0, -1, 0)$, which corresponds to the emergence of the mono-centric pattern (\prettyref{fig:z1}). 
In fact, analogous to \prettyref{eq:omega_2reg}, we can represent $\omega^*$ as follows:  
\begin{align}
    \omega^*
		= 
		\dfrac{\Brx}{\Brv}
			\PDF{(v_1 - v_3)}{x_1} (\BrVtx)
	\label{eq:omega1}, 	
\end{align}
which indicates that moving from region $3$ to $1$ is profitable for workers if $\omega^*$ is positive. 

\begin{figure}
    \centering
    \includegraphics{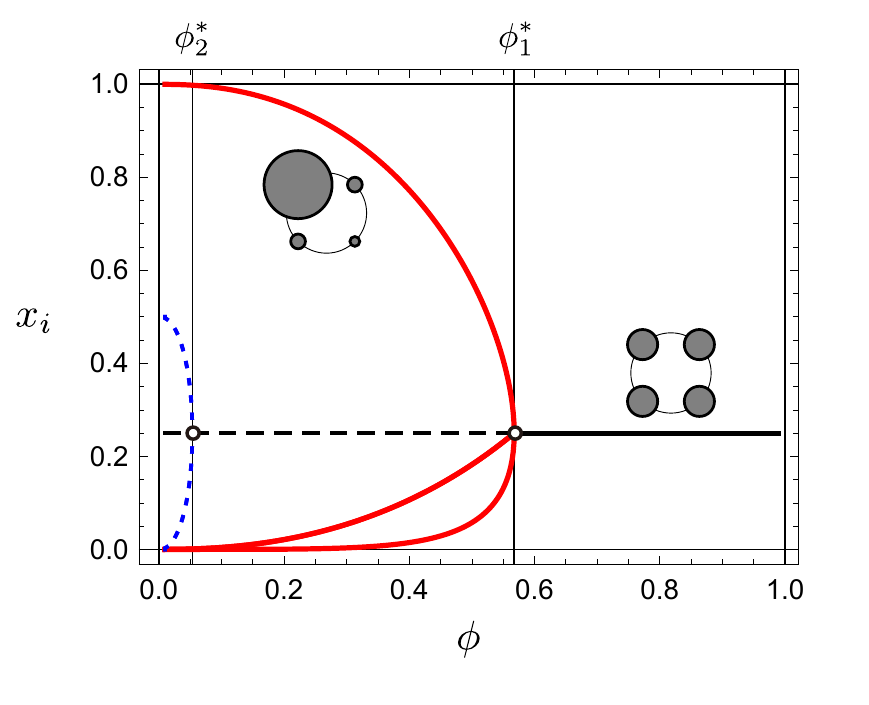}
    \caption{Bifurcation diagram for the four-region circular economy with externality matrix in \prettyref{eq:G-baseline}.}
    \begin{figurenotes}
    $\psi = 0.7$ and $\sigma = 4$. 
    The solid curves indicate stable equilibria and the dashed curves indicate unstable ones. 
    The red curve corresponds to a mono-centric pattern (stable) and the blue curve to a duo-centric one (unstable). The schematics by the solid curves show representative snapshots of the associated stable spatial patterns. 
    \end{figurenotes}
    \label{fig:4-baseline}
\end{figure}
The next proposition characterizes the stability of $\BrVtx$ and the endogenous agglomeration from it. 
\begin{proposition}
\label{prop:4-baseline}
Assume $\VtD$ in \prettyref{eq:D-racetrack} and $\VtG$ in \prettyref{eq:G-baseline}. 
Then, $\omega^* = \Omega(\chi_1,\lambda_1)$ and $\Vtz^* = (1, 0, -1, 0)$, where $\chi_1 = \frac{1 - \phi}{1 + \phi}$ and $\lambda_1 = \frac{1 - \psi}{1 + \psi}$. 
Uniform distribution $\BrVtx$ is linearly stable if and only if $\omega^* < 0$. 
When $\BrVtx$ becomes unstable, then a mono-centric pattern of form $\BrVtx + \epsilon \Vtz^* = (\Brx + \epsilon, \Brx, \Brx - \epsilon, \Brx)$ emerges ($\epsilon > 0$). 
\end{proposition}
As already seen, $\omega^*$ is represented by $\Omega$ in \prettyref{eq:Omega}. 

\prettyref{fig:4-baseline} shows the bifurcation diagram on the $\phi$-axis to numerically confirm \prettyref{prop:4-general}. 
The critical level of $\phi$ at which $\BrVtx$ becomes stable ($\omega^* = 0$) is indicated by $\phi_1^*$. 
When $\phi$ is at its lower extreme ($\phi \rightarrow 0$), workers concentrate in a single region because this way they can avoid the burden of costly transportation altogether \citep{gaspar2020new}; we recall that there are no immobile workers living in any region in the present model. 
The spatial distribution is full agglomeration (e.g., $\Vtx \approx (1,0,0,0)$ as shown in \prettyref{prop:mono-centric}).
As $\phi$ increases, the relative rise of the dispersion force induces a crowding-out from the populated region. 
The spatial pattern becomes, for example, $\Vtx = (x,x',x'',x')$ with $x > x' > x''$, which is still a mono-centric pattern. 
As $\phi$ increases, the spatial pattern gradually flattens and, at threshold $\phi_1^{*}$, the monocentric configuration connects with the uniform distribution. 
If we start from $\BrVtx$ and gradually \emph{decrease} $\phi$ to determine the reverse-reproduced dispersion process, at $\phi_1^*$ the spatial pattern must deviate in the direction of the ``formation'' of a mono-centric configuration (\prettyref{fig:z1}). 
On the left-hand side of the figure, the two city pattern also emerges (indicated by the blue dashed curve) from $\BrVtx$ at $\phi = \phi_2^*$. 
However, this configuration is always unstable in the presented example.

\subsection{The timing of dispersion}
\label{sec:4-equidistant}

We next illustrate that the timing of the agglomeration varies with the structure of $\VtG$. 
If the level of inter-regional production externalities is the same across all region pairs, we can assume $\VtG$ takes the following form, which is a special case of \prettyref{fig:g2} with $\psi' = \psi$:
\begin{align}
    \VtG  
	=
	\begin{bmatrix}
		1 & \psi & \psi & \psi\\
		 & 1 & \psi & \psi\\ 
		 &  & 1 & \psi\\
		\multicolumn{2}{c}{\text{\emph{Sym.}}} & & 1
	\end{bmatrix}. 
	\label{eq:G-equidistant}
\end{align}
This setup is akin to the equidistant geographical networks (see \prettyref{remark:equi}). 
This economy can be thought of as an ``almost connected economy,'' since the payoff in a region is invariant under the permutation of mobile workers in the other regions. That is, for a region's productivity, the exact distribution of workers over the other regions does not matter. 

We obtain the following result. 
\begin{proposition}
\label{prop:4-equidistant}
Assume $\VtD$ in \prettyref{eq:D-racetrack} and $\VtG$ in \prettyref{eq:G-equidistant}. 
Then, $\omega^* = \Omega(\chi_1,\lambda_1')$ and $\Vtz^* = (1, 0, -1, 0)$, where $\chi_1 = \frac{1 - \phi}{1 + \phi}$ and $\lambda_1' \Is \frac{1 - \psi}{1 + 3\psi}$. 
Uniform distribution $\BrVtx$ is linearly stable if and only if $\omega^* < 0$. 
When $\BrVtx$ becomes unstable, a mono-centric pattern of form $\BrVtx + \epsilon \Vtz^* = (\Brx + \epsilon, \Brx, \Brx - \epsilon, \Brx)$ emerges ($\epsilon > 0$). 
\end{proposition}
\noindent 
We have $\Vtz^* = (1, 0, -1, 0)$, which is the same as \prettyref{prop:4-baseline}. 
The only difference is that we have
\begin{align}
    \lambda_1' = \frac{1 - \psi}{1 + 3\psi}
\end{align}
instead of $\lambda_1$. 
Similar to $\lambda$ in the two-region case, $\lambda_1'$ is the eigenvalue of \prettyref{eq:G-equidistant} associated with $\Vtz^*$ and represents the \emph{sensitivity} of the regional productivity when $\BrVtx$ is perturbed by migration. 
We have 
\begin{align}
    \lambda_1' 
    = \frac{1 - \psi}{1 + 3\psi}
    < \frac{1 - \psi}{1 + \psi}
    = \lambda_1 
\end{align}
where $\lambda_1$ is the corresponding sensitivity in the baseline case (\prettyref{sec:4-baseline}). 
That is, productivity gains due to migration are higher for the network in \prettyref{eq:G-baseline} than for that in \prettyref{eq:G-equidistant}; 
the regions in the latter network are more ``connected'' than those in the former and thus migration (or forming agglomeration) is less profitable. 

As we observe $\Omega(s,t)$ is increasing in $t$, we have $\Omega(\chi_1,\lambda_1) > \Omega(\chi_1,\lambda_1')$; 
this means that $\omega^*$ in \prettyref{prop:4-baseline} is always greater than that in \prettyref{prop:4-equidistant}. 
As a result, $\BrVtx$ is stable for a smaller range of $(\phi,\psi)$ in the former than in the latter. 
\prettyref{fig:stab-unif-4_equi} illustrates this observation. 
The solid and dashed curves respectively indicate critical pairs $(\phi^*,\psi^*)$ for the externality matrices \prettyref{eq:G-equidistant} and \prettyref{eq:G-baseline}. 
For each case, $\BrVtx$ is stable in the region above the threshold curve. 
The solid curve is always below the dashed curve, so that $\BrVtx$ is stable for a broader range of $\phi$ and $\psi$ when we assume $\VtG$ in \prettyref{eq:G-equidistant}. 
This example shows that, in a more connected economy, there is less incentive to form an agglomeration. 

If production externalities are governed solely by, for example, internet communications, then whether workers are in the same region or not is less important than the baseline case. 
The externality matrix in \prettyref{eq:G-equidistant} represents this situation. 
Therefore, as \cite{cairncross1997death} argued, the advancement of information technology may indeed encourage dispersion.

\begin{figure}[tb]
	\centering
	    \includegraphics{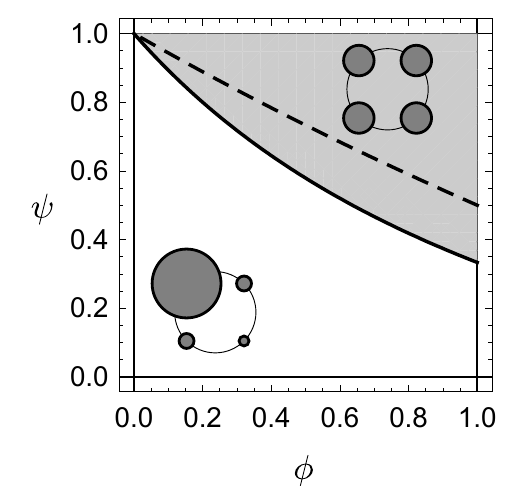}
	\caption{Stability of $\BrVtx$ in four-region economies with externality matrices as in \prettyref{eq:G-baseline} and \prettyref{eq:G-equidistant} ($\sigma = 4$).}
	\begin{figurenotes}
	The black solid curve shows critical pairs $(\phi^*,\psi^*)$ at which $\BrVtx$ becomes unstable for the case  $\VtG = \VtG_{g}$ with $\psi' = \psi$. The dashed curve corresponds to the baseline case \prettyref{eq:G-baseline}. Uniform distribution $\BrVtx$ is stable for the regions above these curves, where the gray regions correspond to $\VtG = \VtG_e$. The solid curve stays below the dashed curve, that is, $\BrVtx$ is stable for a wider range of $(\phi,\psi)$ when the economy is more connected. 
	\end{figurenotes}
	\label{fig:stab-unif-4_equi}
\end{figure}

\subsection{The form of dispersion: Super-regions}
\label{sec:4-super-regions}

Assume the pairs of regions $\{1,2\}$ and $\{3,4\}$ are ``super-regions,'' 
in that production externalities between the regions in the same super-region are stronger than those between two regions in different super-regions (\prettyref{fig:g3}). 
This structure can be represented by the following externality matrix: 
\begin{align}
\VtG
	= 
	\begin{bmatrix}
	1 & \psi' \\
	\psi' & 1
	\end{bmatrix}
	\otimes
	\begin{bmatrix}
	1 & \psi \\
	\psi & 1
	\end{bmatrix}
	=
	\begin{bmatrix}
		\ 1 & \psi & \psi'\psi & \psi'\\
		 & 1 & \psi' & \psi'\psi\\ 
		 & & 1 & \psi \\
		\SYM   & & 1 
	\end{bmatrix}, 
	\quad  
	\label{eq:G-block}
\end{align} 
where we assume $\psi' < \psi$ without loss of generality.

The following proposition shows that the stability of $\BrVtx$ is determined by the magnitude of the externalities between super-regions, $\psi'$, and the bifurcation from $\BrVtx$ leads to the formation of a North--South pattern (\prettyref{fig:z3}). 
\begin{proposition}
\label{prop:4-super-regions}
Assume $\VtD$ in \prettyref{eq:D-racetrack} and $\VtG$ in \prettyref{eq:G-block}. 
Then, $\omega^* = \Omega(\chi,\lambda')$ and $\Vtz^* = (1, 1, -1, -1)$, where $\chi = \frac{1 - \phi}{1 + \phi}$ and $\lambda' \Is \frac{1 - \psi'}{1 + \psi'}$. 
Uniform distribution $\BrVtx$ is linearly stable if and only if $\omega^* < 0$. 
When $\BrVtx$ becomes unstable, then a North--South pattern of form $\BrVtx + \epsilon \Vtz^* = (\Brx + \epsilon, \Brx + \epsilon, \Brx - \epsilon, \Brx - \epsilon)$ emerges ($\epsilon > 0$). 
\end{proposition}

\begin{figure}
    \centering
    \includegraphics{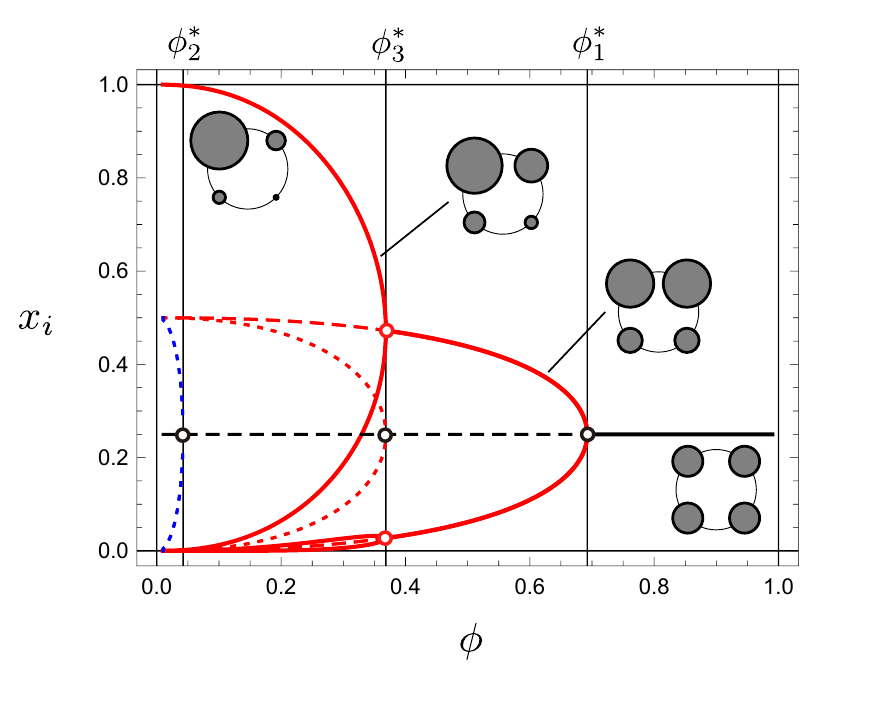}
    \caption{Bifurcation diagram for the four-region circular economy with externality matrix in \prettyref{eq:G-block}.}
	\begin{figurenotes}
	$\psi' = \psi^2 < \psi$ with $\psi = 0.8$ and $\sigma = 4$. The solid curves indicate stable equilibria and the dashed or dotted curves indicate unstable ones. The red curves indicate mono-centric patterns and the blue curve represents a duo-centric one. 
	The schematics by the solid curves show representative snapshots of associated stable spatial patterns. 
	For $\phi \in (\phi_3^*, \phi_2^*)$, a North--South pattern emerges. 
	The economy exhibits a hierarchical structure of North--South asymmetry and intra-regional asymmetries for range $\phi \in (0,\phi_3^*)$. 
	\end{figurenotes}
    \label{fig:four-regions_bloc}
\end{figure}

That is, a four-region economy with $\VtG$ in \prettyref{eq:G-block} has similar properties to the two-region case. 
The intra-super-region interaction level, $\psi$, does not affect the stability of $\BrVtx$, because $\omega^*$ does not include it. 
Each super-region can be regarded as a ``big'' region, so we recover the two-region economy. 

There are two other possible migration patterns, but they are less desirable for mobile workers than the North--South pattern $\Vtz^*$. 
Specifically, there are three possible migration patterns in this economy: 
\begin{align}
    \Vtz_1 = (1, 1,-1, -1),
    \quad 
    \Vtz_2 = (1,-1,1,-1), 
    \quad 
    \text{and}
    \quad
    \Vtz_3 = (1,-1,-1,1),
\end{align}
where $\Vtz^* = \Vtz_1$. 
However, \prettyref{app:stability} shows that the gains (i.e., the eigenvalues of $\VtV$) associated with these patterns satisfy $\omega^* = \max\{\omega_1,\omega_2,\omega_3\} = \omega_1$ for all $(\phi,\psi,\sigma)$ provided that \emph{inter}-super-region externalities are smaller than the \emph{intra}-super-region ones ($\psi' < \psi$). 
The second migration pattern, $\Vtz_2$, is the duo-centric pattern (\prettyref{fig:z2}), whereas the third is an ``East--West'' pattern (90$^\circ$ rotation of the North--South pattern). 
The East--West pattern is more desirable than the duo-centric pattern, because, for the former, the two big regions are close to each other and hence enjoy greater productivity than the latter. 
Similarly, the North--South pattern benefits from greater productivity than the East--West pattern, since the production externalities between the two big regions are $\psi$ in the former and $\psi' < \psi$ in the latter. 
Therefore, the most profitable deviation from $\BrVtx$ is the North--South pattern $\Vtz_1$. 

\prettyref{fig:four-regions_bloc} shows a numerical example for this case. 
Each $\phi^*_k$ indicates the level of $\phi$ for which we have $\omega_k = 0$ ($k = 1,2,3$). 
When $\phi$ is small, a one-peak distribution is stable (solid red curve). 
As $\phi$ increases, dispersion proceeds. 
The difference from \prettyref{fig:4-baseline} is that the North regions attract consistently more workers than the South ones. 
Range $\phi\in(0,\phi_3^*)$ represents the combined process of gradual dispersion from the North to the South and that within each super-region. 
For the South, the process is ambiguous because there are two effects (within- and between-super-regions) at work. 
For range $\phi\in(\phi_3^*,\phi_1^*)$, there is a stable North--South pattern in which the regions in each super-region have the same size; 
this configuration connects smoothly to $\BrVtx$ at critical value $\phi_1^*$, as predicted by \prettyref{prop:4-super-regions}. 
The process is understood as a hierarchical combination of the two-region case (\prettyref{sec:two-region_economy}). 

The economy naturally converges to baseline case  (\prettyref{sec:4-baseline}) as $\psi'\to\psi$. 
We have $\phi_3^* \to \phi_1^*$ as $\psi' \to \psi$ and then the bifurcation at $\phi_1^* = \phi_3^*$ lead to the mono-centric pattern (\prettyref{fig:z1}). 
In fact, when $\Vtz_1$ and $\Vtz_3$ become profitable for workers at the same time, the migration pattern becomes $\frac{1}{2}(\Vtz_1 + \Vtz_3) = (1,0,-1,0)$, that is, the mono-centric pattern in \prettyref{fig:4-baseline}.

\subsection{The form of dispersion: Bypasses}
\label{sec:Gg}

In all cases considered in Sections \ref{sec:4-baseline}, \ref{sec:4-equidistant}, and \ref{sec:4-super-regions}, endogenous mechanisms induce a mono-centric agglomeration, where one of the regions (or a pair of neighboring regions) becomes the center of the economy. 
But other endogenous spatial patterns may arise, in fact, depending on the interaction structure $\VtG$.  
To illustrate this, we consider the following setting (\prettyref{fig:g2}): 
\begin{align}
	& 
	\VtG
	=
	\begin{bmatrix}
		1 & \psi & \psi' & \psi\\
		 & 1 & \psi & \psi'\\ 
		 &   & 1 & \psi\\
		\SYM   &   & 1
	\end{bmatrix}, 
	\label{eq:G-bypass}
\end{align}
where we require $\psi' > 2\psi - 1$ to satisfy \prettyref{assum:a_i} (b). 

In this case, we have the following characterization, which includes Propositions \ref{prop:4-baseline} and \ref{prop:4-equidistant} as special cases where we respectively set $\psi' = \psi^2$ and $\psi' = \psi$. 
\begin{proposition}
\label{prop:4-general}
Assume $\VtD$ in \prettyref{eq:D-racetrack} and $\VtG$ in \prettyref{eq:G-bypass} with arbitrary $\psi'\in(0,1)$.
Then, either $\Vtz^* = \Vtz_1 \Is (1,0,-1,0)$ or $\Vtz^* = \Vtz_2 \Is (1,-1,1,-1)$, 
and $\omega^* = \max\{\omega_1,\omega_2\}$,  
with $\omega_k= \Omega(\chi_k, \lambda_k)$ where 
\begin{align}
    \chi_1 \Is \frac{1 - \phi}{1 + \phi}, 
    \chi_2 \Is \left(\frac{1 - \phi}{1 + \phi}\right)^2, 
    \lambda_1 
    \Is 
    \dfrac{1 - \psi'}{1 + 2 \psi + \psi'}, 
    \text{ and }
    \lambda_2
    \Is 
    \dfrac{1 - 2 \psi + \psi'}{1 + 2 \psi + \psi'}. 
\end{align}
Uniform distribution $\BrVtx$ is linearly stable if and only if $\omega^* < 0$. 
When $\BrVtx$ become unstable, 
\begin{enumerate}
	\item a mono-centric pattern $\BrVtx + \epsilon \Vtz_1 = (\Brx + \epsilon, \Brx, \Brx - \epsilon, \Brx)$ emerges if $\omega^* = \omega_1$; or, 
	\item a duo-centric pattern $\BrVtx + \epsilon \Vtz_2 = (\Brx + \epsilon, \Brx - \epsilon, \Brx + \epsilon, \Brx - \epsilon)$ emerges if $\omega^* = \omega_2$.  
\end{enumerate}
\end{proposition}

In contrast to Sections \prettyref{sec:4-baseline} and \prettyref{sec:4-equidistant}, the duo-centric pattern (\prettyref{fig:z2}) can emerge if it is more profitable for workers than the mono-centric pattern. 

\begin{figure}
    \centering
    \includegraphics{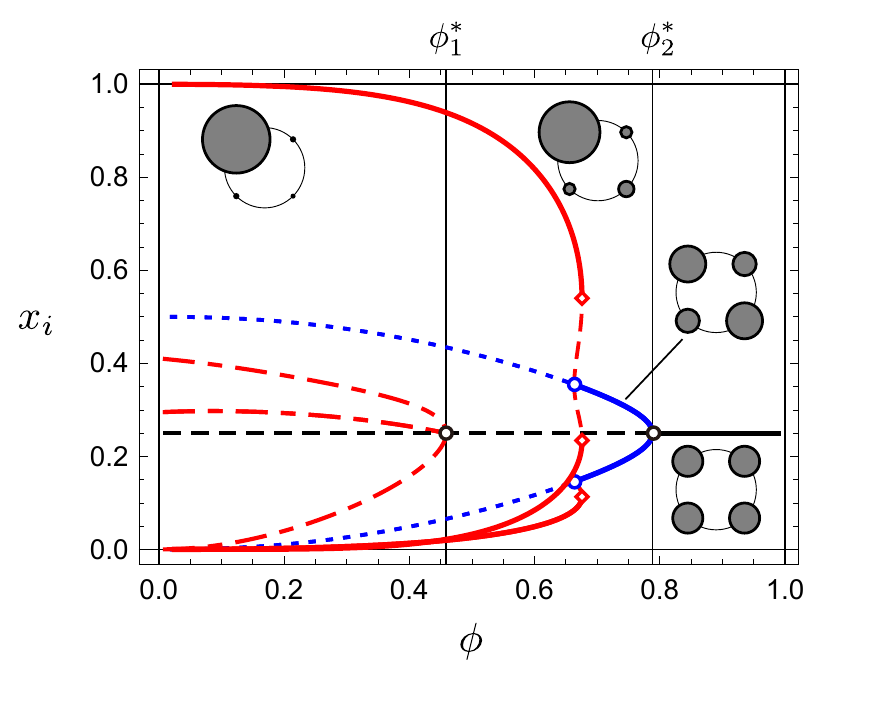}
    \caption{Bifurcation diagram for the four-region circular economy with externality matrix in \prettyref{eq:G-bypass}.}
	\begin{figurenotes}
	$\psi' = \psi^{1/2} > \psi$ with $\psi^{1/2} = 0.65$ and $\sigma = 4$. The solid curves indicate stable equilibria and the dashed or dotted curves indicate unstable ones. The red curves represent mono-centric patterns and the blue curve represents duo-centric patterns. The schematics by the solid curves show representative snapshots of the associated stable spatial patterns. 
	The diamonds ($\lozenge$)  indicate the limit point for the mono-centric configuration drawn by the solid red curves. 
	\end{figurenotes}
    \label{fig:four-regions_duo}
\end{figure}

A necessary condition for the emergence of a duo-centric pattern is $\psi' > \psi$. 
When $\psi' > \psi$, region $1$ is more ``connected'' to region $3$ than to regions $2$ or $4$ in terms of production externalities. 
The regions at antipodal locations in the circle are closer with respect to $\VtG$. 

Similar to Sections \ref{sec:4-baseline} and \ref{sec:4-equidistant}, $\lambda_1$ and $\lambda_2$ represent the marginal productivity gains respectively induced by the formation of mono-centric and duo-centric configurations. 
By employing the formula for $\omega^*$, we show that the uniform distribution can be stable for some $\phi\in(0,1)$ if and only if $(\sigma - 1)\max\{\lambda_1,\lambda_2\} < 1$. 
If inter-regional productivity spillover does not exist (i.e., $\psi = 0$ and $\psi' = 0$), then $\lambda_1 = 1$ and $\lambda_2 = 1$. 
Therefore, $\BrVtx$ cannot be stable for such a case if $\sigma > 2$, which is obviously satisfied by the standard values of $\sigma$ \citep{Anderson-Wincoop-JEL2004}. 
When $\sigma$ is sufficiently large, workers are better off concentrating on a single region because differentiated goods are substitutes whereas inter-regional production externalities are too costly. 

\prettyref{fig:four-regions_duo} shows a numerical example in which $\omega^* = \omega_2$, so that stable duo-centric patterns emerge from $\BrVtx$. 
In line with \prettyref{fig:4-baseline}, $\phi^*_1$ and $\phi^*_2$ respectively indicate the levels of $\phi$ at which we have $\omega_1 = 0$ and $\omega_2 = 0$.  
When $\phi$ is small, a mono-centric distribution is stable (solid red curve), which is similar to \prettyref{fig:4-baseline}. 
As $\phi$ increases, dispersion proceeds. 
The difference from \prettyref{fig:4-baseline} is that region $3$ attracts more workers than regions $2$ and $4$. 
For some range of $\phi$, a duo-centric concentration towards regions $1$ and $3$ emerges, but with asymmetry so that $x_1 > x_3$. 
At some point, the economy jumps to the symmetric two-peaked distribution (solid blue curve) through a saddle-node bifurcation. 
This jump is encountered at a limit point for the mono-centric configuration (indicated by $\lozenge$).  
The two-peaked distribution connects smoothly to the uniform distribution at critical value $\phi_2^*$, which is predicted by \prettyref{prop:4-general}.

\section{Concluding remarks}
\label{sec:conclusion}

This paper proposed a bare-bones general equilibrium model with two proximity structures, one due to trade linkages and the other due to productivity spillovers. 
The former is the standard on which many economic geography models in the literature focus. 
For a symmetric two-region economy, we confirm that the uniform dispersion is stable when the economy is more integrated. 

In multi-region settings, there are two first-order theoretical insights into the role of an additional proximity structure due to production externalities. 
First, we demonstrated that the structure of the externality matrix affects the timing of endogenous agglomeration. 
Particularly, when the economy is tightly connected, there are less incentives for workers to form agglomerations. 
Second, the structure of inter-regional production externalities, even when ex-ante symmetric, can endogenously determine the spatial distribution of workers across regions. 
Our examples, for instance, show that the number of economic centers can depend on the network structure of production externalities. 
Although the economy collapses to the flat-earth pattern in the frictionless limit, the structure of social proximity emerges as a determinant of the geographical distribution of workers when transportation costs are at intermediate levels.

The model in this paper is a simple thought experiment, as we build on the compromise that the externality structure is exogenously given. 
As discussed in \prettyref{sec:introduction}, there are various interpretations of the externality matrix.
The most interesting extension would be its endogenous determination. 
For such micro-founded models, our framework can be utilized to obtain coarse insights into the role of an additional network structure.

\clearpage 
\appendix
\small\singlespacing

\section{Proofs}
\label{app:proofs}
\numberwithin{assumption}{section}
\numberwithin{lemma}{section}
\numberwithin{proposition}{section}
\numberwithin{fact}{section}

This appendix collects the omitted proofs and derivations. 
Throughout, for a vector-valued function $\Vtf$, we denote $\VtF_x \Is [\partial f_i/\partial x_j]$, and at which point it is evaluated is inferred from the context.  
For vector $\Vtx$, we denote $\HtVtx \Is \diag{\Vtx}$. 

\subsection{General properties} 

For any $\Vtx > \Vt0$, there is a unique wage vector that solves the market equilibrium conditions. 
\begin{lemma}
\label{lem:w-uniqueness}
There exists a unique positive solution $\Vtw$ to \prettyref{eq:short-run-equilibrium} and \prettyref{eq:normalization} for any positive spatial distribution $\Vtx$. 
Further, for any $i\inN$, $w_i \to \infty$ as $x_i \to 0$.  
\end{lemma}
\begin{proof}[Proof of \prettyref{lem:w-uniqueness}]
\textit{Existence}. 
Let the excess demand function $\VtPsi(\Vtw) = (\Psi_i(\Vtw))_{i\inN}$ be defined as follows: 
\begin{align}
&
    \Psi_i(\Vtw) \Is 
            \frac{1}{w_i}
            \left(
	        \sum_{k\inN}
				\dfrac
				{a_i^{\sigma - 1} w_i^{1 - \sigma} \phi_{ik}}
				{\sum_{l\inN} a_l^{\sigma - 1} w_l^{1 - \sigma} \phi_{lk}}
			w_k x_k
			- 
		    w_i x_i
		    \right). 
\end{align}
Then, market wage is a solution for $\VtPsi(\Vtw) = \Vt0$. 
We see (i) $\VtPsi(\cdot)$ is continuous, (ii) $\VtPsi(\cdot)$ is homogeneous of degree zero, (iii) $\Vtw^\top\VtPsi(\Vtw) = 0$ for any $\Vtw$ because of normalization of world income, and (iv) $\Psi_i(\Vtw) > - x_i > - 1$. 
Additionally, (v) for any sequence $\{\Vtw^n\}_{n = 0}^\infty$ of strictly positive $\Vtw^n = (w_i^n)_{i\inN}$ that converges to some $\TlVtw = (\Tlw_i)_{i\inN}$ such that $\Tlw_{i^*} = 0$ for some $i^*\inN$, we have $\max_{i\inN} \Psi_i(\Vtw^n) \to \infty$ as $n\to\infty$,  because we see
\begin{align}
    \max_{i\inN} 
        \Psi_i(\Vtw)
        & 
        = 
        \max_{i\inN}
	        \sum_{k\inN}
				\dfrac
				{a_i^{\sigma - 1} w_i^{- \sigma} \phi_{ik}}
				{\sum_{l\inN} a_l^{\sigma - 1} w_l^{1 - \sigma} \phi_{lk}}
			w_k x_k
			- 
		\max_{i\inN}
			x_i
		\\ 
		& 
		> 
        \max_{i\inN}
        \max_{j\inN}
				\dfrac
				{a_i^{\sigma - 1} w_i^{-\sigma} \phi_{ij}}
				{\sum_{l\inN} a_l^{\sigma - 1} w_l^{1 - \sigma} \phi_{lj}}
			w_j x_j
			- 
		    1
		\label{eq:j_ind}
		\\
		& 
		> 
		    \max_{i\inN}
				\dfrac
				{a_i^{\sigma - 1} \phi_{ij^*}}
				{\sum_{l\inN} a_l^{\sigma - 1} \phi_{lj^*}}
				w_{j^*} x_{j^*}
				\frac{w_i^{-\sigma}}{(\textstyle{\min_{k\inN} w_k})^{1 - \sigma}}
			- 
		    1 
		\\
		& 
		> 
		    \max_{i\inN}
				\dfrac
				{a_i^{\sigma - 1} \phi_{ij^*}}
				{\sum_{l\inN} a_l^{\sigma - 1} \phi_{lj^*}}
				w_{j^*} x_{j^*} 
				\frac{1}{\textstyle{\min_{k\inN} w_k}}
			- 
		    1, 
		\label{eq:max_Psi_bound}
\end{align}
where $j^*\inN$ is the regional index that achieves the second maximum in \prettyref{eq:j_ind}. 
The right hand side of the last display \prettyref{eq:max_Psi_bound} goes to positive infinity because $\min_{k\inN} w_k^n \to \min_{k\inN} \Tlw_k = \Tlw_{i^*} = 0$ as $n \to \infty $ and the other component is positive. 
Therefore, $\VtPsi(\Vtw)$ satisfies hypothesis (i)--(v) of \cite{MWG}, Proposition 17.C.1 on p.585 and there is $\Vtw$ such that $\VtPsi(\Vtw) = \Vt0$ and $w_i > 0$ for all $i\inN$. 

\textit{Uniqueness}. 
Note that $\VtPsi(\cdot)$ has the \textit{gross substitute property}. 
That is, 
\begin{align}
    \PDF{\Psi_i(\Vtw)}{w_j} 
    = 
    \frac{1}{w_i}
    \left(
        m_{ij} x_j + (\sigma - 1) \sum_{k\inN} m_{ik} m_{jk} w_k x_k \frac{1}{w_j}
    \right)
    > 0
\end{align}
for any $i \ne j$, 
where we let
\begin{align}
    m_{ij} = 
		\frac
		{a_i^{\sigma - 1} w_i^{1 - \sigma} \phi_{ij}}
		{\sum_{l\inN} a_l^{\sigma - 1} w_l^{1 - \sigma} \phi_{lj}}.
	\label{eq:m_ij}
\end{align}
Note that the first term in the parenthesis is nonnegative and the second term is positive. Hence, by Proposition 17.F.3 of \cite{MWG}, there exists unique $\Vtw$ up to scale such that $\VtPsi(\Vtw) = \Vt0$.  Thus, there is unique $\Vtw$ that satisfy $\VtPsi(\Vtw) = \Vt0$ and $\sum_{i\inN} w_i x_i = 1$. 

\textit{Diverge when $x_i \to 0$}. By rewriting the condition $\Psi_i(\Vtx) = 0$, we see
\begin{align}
    w_i^\sigma = 
    \frac{1}{x_i} 
	    \sum_{k\inN}
			\dfrac
				{a_i^{\sigma - 1} \phi_{ik}}
				{\sum_{l\inN} a_l^{\sigma - 1} w_l^{1 - \sigma} \phi_{lk}}
			w_k x_k
	& 
	> 
    \frac{1}{x_i} 
        \min_{k\inN}
			\dfrac
				{a_i^{\sigma - 1} \phi_{ik}}
				{\sum_{l\inN} a_l^{\sigma - 1} w_l^{1 - \sigma} \phi_{lk}}
		\sum_{j\inK}
			w_j x_j
	\\
	& 
	> 
    \frac{1}{x_i} 
        \min_{k\inN}
			\dfrac
				{a_i^{\sigma - 1} \phi_{ik}}
				{\sum_{l\inN} a_l^{\sigma - 1}\phi_{lk}}
		\min_{l\inN}
		w_l^{\sigma - 1}
	= 
	C_1 C_2 x_i^{-1}, 
    \label{eq:wi_bound}
\end{align}
where $C_1 \Is \min_{k\inN} a_i^{\sigma - 1} \phi_{ik} \left(\sum_{l\inN} a_l^{\sigma - 1} \phi_{lk}\right)^{-1} > 0$ and $C_2 \Is \min_{l\inN} w_l^{\sigma - 1} > 0$, that is, $w_i\to\infty$ as $x_i \to 0$. 
\end{proof}

\begin{proof}[Proof of \prettyref{prop:existence}]
In general, a spatial equilibrium is a spatial distribution $\Vtx\inX$ such that the following Nash equilibrium condition is met for the location choice of workers:
\begin{align}
	\label{eq:equilibrium}
	\left\{
	\begin{array}{l}
		\text{$v^* = v_i(\Vtx)$
		for all regions $i\inN$ with $x_i > 0$,}\\
		\text{$v^* \geq v_i(\Vtx)$
		for any region $i\inN$ with $x_i = 0$,}
	\end{array}\right.
\end{align} 
where $v^*$ is an equilibrium payoff. 
The following \textit{variational inequality problem} is equivalent to \prettyref{eq:equilibrium}:
\begin{align}
    & \text{Find $\Vtx\inX$ such that $\Vtv(\Vtx)^\top(\Vty - \Vtx) \le 0$ for all $\Vty \inX$}. 
    \tag{VIP}
    \label{eq:VIP0}
\end{align}

Every region is necessarily populated at any spatial equilibrium when $\VtD$ and $\VtG$ are positive (i.e., $\phi_{ij} > 0$ and $\psi_{ij} > 0$ for all $i,j\inN$). 
First, $a_i(\Vtx) > 0$ for all $i\inN$ at any $\Vtx\inX$ when $\VtG$ is positive. 
Next, we see
\begin{align}
    v_i^{\sigma - 1} 
    & 
    = 
    w_i^{\sigma - 1} 
	    \sum_{j\inN}
		    a_j^{\sigma - 1} w_j^{1 - \sigma} \phi_{ji}
	> 
	w_i^{\sigma - 1}
	    \sum_{j\ne i}
		    a_j^{\sigma - 1} w_j^{1 - \sigma} \phi_{ji}
    > 
    C_3 C_4 w_i^{\sigma - 1}
    \label{eq:vi_bound}
\end{align}
where 
$C_3 \Is \min_{j\ne i}a_j^{\sigma-1} \phi_{ij} > 0$ (by positivity of $\phi_{ij}$) 
and $C_4 \Is \min_{j\ne i} w_j^{1 - \sigma} > 0$ (by positivity of $w_i$ shown by \prettyref{lem:w-uniqueness}). 
Note also that $C_4$ is bounded for all positive $\Vtx$. 

By \prettyref{eq:vi_bound} and \prettyref{eq:wi_bound}, for any sequence of positive spatial patterns $\{\Vtx^n\}_{n = 1}^\infty$ that converges to a spatial distribution such that $x_i = 0$, $w_i(\Vtx^n)$ and $v_i(\Vtx^n)$ both diverge to positive infinity as $n\to\infty$. 
On the other hand, by \prettyref{lem:w-uniqueness}, $\Vtw$ is uniquely given if we focus on the regions with positive population $\ClN_+ \Is \{k\inN \mid x_k > 0\}$ by letting $w_i := \infty$ and $w_i x_i:= 0$ for all $i\inN_0 \Is \{k\inN\mid x_k = 0\}$. Then, $v_i(\Vtx)$ is finite for all $i\inN_+$, while $v_i(\Vtx)$ is infinitely large for any $i\inN_0$. 
Since such spatial distribution cannot be a spatial equilibrium, every region is necessarily populated in equilibrium. 
Thus, the equilibrium condition \prettyref{eq:equilibrium} in fact reduces to the equality:  $v_i(\Vtx) = v_j(\Vtx)$ for all $i,j\inN$. 

Consider the following variational inequality problem:
\begin{align}
    & \text{Find $\Vtx\inX_\epsilon$ such that $\Vtv(\Vtx)^\top(\Vty - \Vtx) \le 0$ for all $\Vty \inX_\epsilon$},  
    \tag{VIP$\epsilon$}
    \label{eq:VIP}
\end{align}
where $\ClX_\epsilon \Is \{\Vtx\inX \mid x_i > \epsilon\ \forall i\inN\}$ for some $\epsilon > 0$. 
Since $\Vtv$ is differentiable and thus continuous on $\ClX_\epsilon$, and $\ClX_\epsilon$ is compact and convex,  
by Corollary 2.2.5 of \cite{Facchinei-Pang-Book2003}, the set of solutions for \prettyref{eq:VIP} is nonempty and compact for any choice of $\epsilon > 0$. 
There is some $\epsilon > 0$ for which all solutions for \prettyref{eq:VIP} are in the (relative) interior of $\ClX_\epsilon$ because $v_i(\Vtx)$ is continuous in $\Vtx$ and diverges when $x_i \to 0$. 
Because any interior solution for \prettyref{eq:VIP} must satisfy $v_i(\Vtx) = v_j(\Vtx)$ for all $i,j\inN$,  they are spatial equilibria. 
\end{proof}

\begin{proof}[Proof of \prettyref{prop:dispersion}]
Suppose $\phi_{ij} = 1$ and $\psi_{ij} = 1$ for all $i,j\inN$. 
Then, $a_i(\Vtx) = \sum_{j\inN} x_j = 1$ and  
\begin{align}
    m_{ij} = \frac{w_i^{1 - \sigma} \phi_{ij}}{\sum_{l\inN} w_l^{1- \sigma} \phi_{lj}} = m_i \Is \frac{w_i^{1 - \sigma}}{\sum_{l\inN} w_l^{1- \sigma}} \in (0,1) 
    \label{eq:m_i}, 
\end{align}
which implies that 
$w_i x_i = \sum_{j\inN} m_{ij} w_{j} x_j = m_i \sum_{j\inN} w_j x_j = m_i$. 
Additionally, we see 
\begin{align}
    v_i(\Vtx) 
    = \frac{w_i}{P_i} 
    = 
            \left(
        \frac{w_i^{1 - \sigma}}
        {
            \sum_{j\inN} w_j^{1 - \sigma}
        }
        \right)^{\frac{1}{1 - \sigma}} 
    = 
        m_i^{\frac{1}{1 - \sigma}}. 
\end{align}
Because all regions are populated in equilibrium, $v_i(\Vtx) = v_j(\Vtx)$ for all $i\inN$, which gives $m_i = \Brm \in(0,1)$ for all $i\inN$. 
Then, $w_i = \Brw > 0$ for all $i\inN$ by \prettyref{eq:m_i}. 
Since $w_i x_i = \Brw x_i = m_i = \Brm$, $x_i = \frac{\Brm}{\Brw} = \Brx = \frac{1}{n}$ for all $i\inN$. That is, $\BrVtx$ is the unique equilibrium. 
The stability of $\BrVtx$ can be shown by knowing $\VtV_x$ is negative definite at $\Vtx = \BrVtx$, as $\VtV_x$ is symmetric at $\BrVtx$ and all of its relevant eigenvalues equals $-\sigma^{-1} < 0$. 
\end{proof}

\begin{proof}[Proof of \prettyref{prop:mono-centric}]
When $\phi_{ij} = 0$ and $\psi_{ij} = 0$ so that inter-regional trade and production externalities are prohibitive, we have $a_i = x_i$ for all $i\inN$. 
The market equilibrium conditions reduce to:
$w_i x_i = x_i^{\sigma - 1} w_i^{1 - \sigma} w_i x_i$, 
thereby $w_i = x_i$. 
Thus, $P_i = (x_i^{\sigma - 1} x_i^{1 - \sigma})^{1/(1 - \sigma)} = 1$ and $v_i(\Vtx) = x_i$.  
Therefore, any spatial distribution in which the populated regions have the same population is an equilibrium when $\phi_{ij}$ and $\psi_{ij}$ vanish for all $i \ne j$. 
However, all such equilibria with more than one populated regions cannot be stable under natural dynamics since  
any migration between populated regions induce relative payoff advantages as $v_i(\Vtx) > v_j(\Vtx)$ if $x_i > x_j$. 
Thus, the economy ends up with a full agglomeration in one of the regions. 
\end{proof}

\subsection{Stability of the uniform distribution}
\label{app:stability}

Since multiple equilibria may arise due to the centripetal force embedded by $\Vta(\cdot)$, we consider equilibrium refinement based on some myopic dynamics $\DtVtx = \Vtf(\Vtx)$. 
We focus on dynamics of the form $\DtVtx = \Vtf(\Vtx) = \Vtf(\Vtx,\Vtv(\Vtx))$, that is, the dynamics that maps spatial distribution $\Vtx$ and payoff level $\Vtv(\Vtx)$ to a motion vector. 

We focus on a class of dynamics that satisfy the following assumptions. 
\begin{assumption}
\label{assum:dynamics}
Both $\Vtf$ and $\TlVtf$ are differentiable for all positive $\Vtx$ and satisfy 
(i) $\Vtf(\Vtx) = \Vt0$ if and only if $\Vtx$ is a spatial equilibrium of our model,  
(ii) if $\Vtf(\Vtx)\neq\Vt0$, then $\Vtv(\Vtx) ^\top \Vtf(\Vtx) > 0$, 
and 
(iii) $\VtP\Vtf(\Vtx,\Vtv(\Vtx)) = \Vtf(\VtP\Vtx,\VtP \Vtv(\Vtx))$ for all permutation matrices $\VtP$. 
\end{assumption}
\noindent 
Conditions (i) and (ii) are, respectively, called \textit{Nash stationality} and \textit{positive correlation} \citep{Sandholm-Book2010}, which are the most parsimonious assumptions we can impose on a dynamic $\Vtf$ to be consistent with the underlying model (payoff function) $\Vtv$. 
Condition (iii) ensures that $\Vtf$ is not biased, that is, it does not include any ex-ante preference for some regions to the others. 
In other words, we suppose that all location incentives for workers are captured by the payoff function. 
We suppose $\Vtf$ is $\RmC^1$, only because we employ linear stability as the definition of stability. 
We call dynamics that satisfy \prettyref{assum:dynamics} \textit{admissible dynamics}. 

We first give general characterization of stability of $\BrVtx$ for both $n = 2$ and $n = 4$. 
Our approach build on \cite{Akamatsu-Takayama-Ikeda-JEDC2012}, which is recently synthesized by \cite{AMOT-DP2019}. 

For any $\VtD > 0$ and $\VtG > 0$, we compute that
\begin{align}
    \VtV_x
	=
	\HtVtv(\Vtx)
	\left(
		\VtM^\top
		\HtVta^{-1}
		\VtA_x
		+
		\left(
		\VtI
		-
		\VtM^\top
		\right) 
		\HtVtw^{-1}
		\VtW_x
	\right) 
	\label{eq:Vx_gen}
\end{align}
where 
$\VtM = [m_{ij}]$ with $m_{ij}$ defined by \prettyref{eq:m_ij}. 
The Jacobian matrix of wage with respect to spatial distribution $\Vtx$, $\VtW_x$, is given by the implicit function theorem regarding the short-run market equilibrium condition:
\begin{align}
    z_i(\Vtx,\Vtw) \Is 
        w_i x_i 
        - 
	        \sum_{k\inN}
			m_{ik}
			w_k x_k 
		= 0
\end{align}
as $\VtW_x = - \VtZ_w^{-1}\VtZ_x$. 
We compute 
\begin{align}
    &
    \VtZ_w
	= 
	\left( 
	\sigma \left( \VtI - \VtM \right) \HtVty 
	- 
	\left( \sigma - 1 \right)  \VtM\HtVty \VtM^\top 
	\right) 
	\HtVtw^{-1}, 
	\\
    & 
	\VtZ_x
	= 
	\left(
		\left(\VtI - \VtM\right)\HtVty
		-
		(\sigma - 1)
		\left(
			\HtVty
			-
			\VtM
			\HtVty
			\VtM^\top
		\right) 
		\HtVta^{-1}
		\VtA_x
		\HtVtx
	\right) 
	\HtVtx^{-1} 
\end{align}
where $\Vty = (w_i x_i)_{i\inN}$. 
We note $\VtA_x = \VtG$. 

Consider a $n$-region economy ($n \ge 2$). 
Suppose uniform distribution $\BrVtx = (\Brx,\Brx,\hdots,\Brx)$ with $\Brx = \frac{1}{n}$. 
Let $\BrVtD$ and $\BrVtG$ be the row-normalized versions of $\VtD$ and $\VtG$. 
We have $\VtM = \BrVtD$ and $\HtVta^{-1} \VtA_x = \BrVtG$ at $\BrVtx$. 
We impose the following assumption on $\BrVtD$ and $\BrVtG$, which is satisfied by our examples in Sections \ref{sec:two-region_economy} and \ref{sec:four-region_economies}. 
\begin{assumption}
\label{assum:circulants}
Both $\BrVtD$ and $\BrVtG$ are either circulants or block circulants with circulant blocks (BCCBs). 
\end{assumption}
Under the assumption, $\BrVtD$ and $\BrVtG$ commute. 
In turn, we evaluate as follows:
\begin{align}
    &
    \VtW_x 
    = 
    - 
	\frac{\Brw}{\Brx}
	\left(
	    \sigma \VtI 
	    + 
	    (\sigma - 1) \BrVtD
	\right)^{-1}
	\left(
	     - \VtI + (\sigma - 1) \left(\VtI + \BrVtD\right) \BrVtG
	\right)
	\label{eq:Wx}
	\\
    & 
	\VtV_x
	= 
	\frac{\Brv}{\Brx}
	\left(
	    \sigma \VtI 
	    + 
	    (\sigma - 1) \BrVtD
	\right)^{-1}
	\left(
	    -(\VtI - \BrVtD)
	    + 
	    ((\sigma - 1)\VtI + \sigma \BrVtD)\BrVtG 
	\right), 
	\label{eq:Vx}
\end{align}
where $\Brv$ is the uniform level of payoff at $\BrVtx$. We let $\VtV = \tfrac{\Brx}{\Brv} \VtV_x$. 
By assumption, $\VtV$ is real and symmetric. 

When $\VtV$ is negative definite with respect to $T\ClX = \{\Vtz\in\BbR^n \mid \sum_{i\inN} z_i = 0\}$, $\BrVtx$ is \textit{evolutionary stable state} \citep[see][Observation 8.3.11]{Sandholm-Book2010}. 
Then, $\BrVtx$ is stable under all admissible dynamics. 
Since $\VtV$ is symmetric, it is negative definite if and only if all of its eigenvalues with respect to $T\ClX$ is negative. 

We have the following fact \citep[see, e.g.,][]{horn2012matrix}. 
\begin{fact}
\label{fact:eigenvalues}
Under \prettyref{assum:circulants}, the eigenvalues of $\VtV$ are given by:
\begin{align}
    & \omega_k = \Omega(\chi_k, \lambda_k) \Is \frac{\Omega^\sharp(\chi_k,\lambda_k)}{\Omega^\flat(\chi_k)}
    & \forall k, 
    \label{eq:omega-general} 
\end{align}
where $\chi_k$ and $\lambda_k$ are, respectively, the $k$th eigenvalues of $\BrVtD$ and $\BrVtG$, and we define 
\begin{align}
    \Omega^\sharp(\chi,\lambda)
    & 
    \Is 
    -
    (1 - \chi) 
    + 
    \left((\sigma - 1) + \sigma\chi \right)
    \lambda,
    \\
    \Omega^\flat(\chi)
    &  
    \Is \sigma + (\sigma - 1)\chi. 
    \label{eq:gain_function}
\end{align}
The three matrices $\VtV$, $\BrVtD$, and $\BrVtG$ share the same set of eigenvectors. 
One eigenvector is $\Vt1 = (1,1,\hdots,1)$ and is orthogonal to $T\ClX$. 
The other eigenvectors $\{\Vtz_k\}$ span $T\ClX$ and each satisfies $\Vtz_k^\top \Vt1 = 0$. 
\end{fact}
The formula \prettyref{eq:gain_function} is simply the translation of the matrix relationship \prettyref{eq:Vx} into an eigenvalue relationship, which made possible by the properties of circulant matrices (or BCCBs).  
In effect, the stability of $\BrVtx$ is determined by eigenvalues $\{\omega_k\}$ that corresponds to the eigenvectors other than $\Vt1$.

Under \prettyref{assum:dynamics}, when some eigenvalues $\omega_k$ switches from negative to positive, then the state is pushed towards the direction of the associated eigenvector \citep[see, e.g., ][Chapter 5]{kuznetsov2013elements}. 
The ``critical'' eigenvector is tangent to the unstable manifold emanating from $\BrVtx$. 
When 
\begin{align}
    \omega_\text{max} \Is \max_k\{ \omega_k \} 
\end{align}
turns its sign from negative to positive, then the spatial distribution is perturbed towards the direction of the associated eigenvector $\Vtz_{k^*}$, where $k^* \Is \argmax_k\{\omega_k\}$. 
Intuitively, each eigenvalues $\omega_k$ is the gain for migrants when they collaterally migrate toward the direction of the associated eigenvector. 

To be more explicit, we proceed as follows.  
Consider a general perturbation of $\BrVtx$ such that $\Vtx = \BrVtx + \Vtz$, where $\Vtz^\top \Vt1 = 0$. 
Rewrite $\Vtz$ as
\begin{align}
    \Vtz = \sum_{k\inK} c_k \Vtz_k, 
    \label{eq:z-decomposition}
\end{align}
where $\ClK \Is \{1,2,3\}$ and $\{\Vtz_k\}_{k\inK}$ are the eigenvectors of $\VtV$ such that $\Vtz_k^\top \Vt1 = 0$ and $\|\Vtz_k\|^2 = \Vtz_k^\top\Vtz_k = 1$. 
For our examples of $\VtD$ and $\VtG$, matrix $\VtV$ is symmetric. 
Thus, the eigendecomposition of $\VtV$ is given by: 
\begin{align}
    \VtV = \omega_0 \Vt1\Vt1^\top + \sum_{k\inK} \omega_k \Vtz_k \Vtz_k^\top, 
    \label{eq:V-decomposition}
\end{align}
where $\omega_k$ is the eigenvalue associated with $\Vtz_k$, with $\Vtz_0 = \Vt1 = (1,1,1,1)$. 
By plugging \prettyref{eq:z-decomposition} and \prettyref{eq:V-decomposition} into \prettyref{eq:baromega}, we can evaluate $\Bromega$ as:
\begin{align}
    \Bromega 
    = 
    \Vtz^\top \VtV \Vtz 
    = 
    \sum_{k \inK} 
        \omega_k 
        c_k^2. 
\end{align}

It shows that $\Bromega$ is maximized by $\Vtz = \Vtz_{k^*}$, where $k^* = \argmax_{k\inK} \omega_k$:   
\begin{align}
    \omega_\text{max} 
    \Is 
    \max_{\Vtz:\|\Vtz\| = 1} \Bromega 
    = 
    \omega_{k^*}, 
\end{align}
where we assume $\|\Vtz\|^2 = \Vtz^\top\Vtz = \sum_{k\inK} c_k^2 = 1$. 
The uniform distribution is stable when
\begin{align}
    \omega_\text{max}  = \omega_{k^*} = \max_{k\inK}\{\omega_k\} < 0. 
\end{align}

The relevant eigenvalues $\{\omega_k\}$ are given by the formula \prettyref{eq:omega-general}. 
They depend on the properties of $\BrVtD$ and $\BrVtG$ through $\{\chi_k\}$ and $\{\lambda_k\}$. 
From \prettyref{eq:omega-general}, the stability of $\BrVtx$ switches when: 
\begin{align}
    \max_{k}\left\{\Omega^\sharp(\chi_k,\lambda_k)\right\}
\end{align}
changes its sign from negative to positive, 
since we have $\Omega^\flat(\chi_k) > 0$ for all examples considered in this paper. 
In the following, we provide the proofs for each examples in Sections \ref{sec:two-region_economy} and \ref{sec:four-region_economies}. 

\begin{proof}[Proof of \prettyref{prop:n=2}]
When $n = 2$, both $\BrVtD$ and $\BrVtG$ are circulants. 
$\VtV$ can be diagonalized by the following discrete Fourier transformation (DFT) matrix: 
\begin{align}
    \VtZ_2 \Is 
    \frac{1}{\sqrt{2}} 
    \begin{bmatrix}
    1 & 1 \\ 1 & - 1
    \end{bmatrix}, 
\end{align}
where each column vector is an eigenvector (migration pattern) of $\VtV$. 
The relevant eigenvector is $\Vtz =(1, -1)$, because $\Vt1 = (1, 1)$ violates the conservation of workers' population. 
The eigenvalue of $\BrVtD$ and $\BrVtG$ associated with $\Vtz$ are, respectively, given by 
\begin{align}
    \chi = \frac{1 - \phi}{1 + \phi}
    \quad\text{and}\quad
    \lambda = \frac{1 - \psi}{1 + \psi}. 
    \label{eq:chi1_lambda1}
\end{align}
That $\omega = \Omega(\chi,\lambda) < 0$ implies $\VtV$ is negative definite with respect to $T\ClX$ and in turn the stability of $\BrVtx$. 
\end{proof}
\begin{proof}[Proof of Propositions \ref{prop:4-baseline}, \ref{prop:4-equidistant}, and \ref{prop:4-general}]
Both $\BrVtD$ and $\BrVtG$ are circulant matrices. 
The externality matrices considered in these propositions are special cases of general specification \prettyref{eq:G-bypass}. 
For these exampels, $\VtV$ is diagonalized by the following DFT matrix: 
\begin{align}
    \VtZ_4 
    \Is 
    \frac{1}{2}
    \begin{bmatrix}
        1 & 1 & 1 & 1\\ 
        1 & \Rmi & -1 & -\Rmi \\
        1 & -1 & 1 & -1\\
        1 & -\Rmi & -1 & \Rmi
    \end{bmatrix}
\end{align}
where $\Rmi$ is the imaginary unit ($\Rmi^2 = -1$). 
The relevant (real) eigenvectors are obtained by combining the columns of $\VtZ_4$ as follows: $\Vtz_1 = (1, 0, -1, 0)$, $\Vtz_2 = (1, -1, 1, -1)$, and $\Vtz_3 = (0, 1, 0, -1)$, where we ommit the normalizing constants for simplicity. 
$\Vtz_1$ and $\Vtz_3$ share the same eigenvalues. This is because $\Vtz_1$ and $\Vtz_3$ represent the same spatial configuration (the economy is symmetric under rotation). 

The eigenvalues of $\BrVtD$ associated with $\Vtz_1$ and $\Vtz_2$ are, respectively, 
\begin{align}
    \chi_1 = \frac{1 - \phi}{1 + \phi}
    \quad\text{and}\quad
    \chi_2 = \left(\frac{1 - \phi}{1 + \phi}\right)^2. 
    \label{eq:chi1-chi2}
\end{align}
Additionally, for the case \prettyref{eq:G-equidistant}, the eigenvalues of $\BrVtG$ associated with $\Vtz_1$ and $\Vtz_2$ are, respectively, 
\begin{align}
    \lambda_1 = \frac{1 - \psi'}{1 + 2\psi + \psi'}
    \quad\text{and}\quad
    \lambda_2 = \frac{1 - 2\psi + \psi'}{1 + 2\psi + \psi'}.  
\end{align}
The uniform distribution is stable when $\omega_k = \Omega(\chi_k, \lambda_k) < 0$ for both $k = 1$ and $k = 2$. 
$\psi' = \psi^2$ implies \prettyref{prop:4-baseline}, 
whereas $\psi' = \psi$ implies \prettyref{prop:4-equidistant}. 
For both the two cases, we have $\omega_1 = \max\{\omega_1,\omega_2\}$, thereby mono-centric pattern $\Vtz_1$ (\prettyref{fig:z1}) emerges from $\BrVtx$. 
\prettyref{prop:4-general} follows because $k = \argmax_k\{\omega_1,\omega_2\}$ depends on the values of $(\phi,\psi,\psi')$. Particularly, $\psi' > \psi$ is necessary for $\omega_2 > \omega_1$. 
\end{proof}

\begin{proof}[Proof of \prettyref{prop:4-super-regions}]
Both $\VtD$ and $\VtG$ are BCCBs and $\VtV$ is diagonalized by the following two-dimensional DFT matrix:
\begin{align}
    \VtZ_2 \otimes \VtZ_2 
    = 
    \frac{1}{\sqrt{2}} 
    \begin{bmatrix}
    \VtZ_2 & \VtZ_2 \\
    \VtZ_2 & - \VtZ_2 
    \end{bmatrix} 
    = 
    \frac{1}{2}
    \begin{bmatrix}
    1 & 1 & 1 & 1\\
    1 & -1 & 1 & -1\\
    1 & 1 & -1 & -1\\
    1 & -1 & -1 & 1
    \end{bmatrix}
\end{align}
The relevant eigenvectors are $\Vtz_1 = (1, 1, -1, -1)$, $\Vtz_2 = (1, -1, 1, -1)$, and $\Vtz_3 = (1, -1, -1, 1)$. 
$\Vtz_1$ is the North--South pattern in \prettyref{fig:z3}, whereas $\Vtz_3$ is its $90^\circ$ rotation (the ``East--West'' pattern) and 
$\Vtz_2$ is the duo-centric pattern. 

The eigenvalues of $\BrVtD$ associated with $\Vtz_1$, $\Vtz_2$, and $\Vtz_3$ are, respectively, 
\begin{align}
    \chi_1 = \frac{1 - \phi}{1 + \phi},
    \quad 
    \chi_2 = \left(\frac{1 - \phi}{1 + \phi}\right)^2, 
    \quad\text{and}\quad
    \chi_3 = \frac{1 - \phi}{1 + \phi}. 
    \label{eq:chi1-chi2-chi3}
\end{align} 
For the case \prettyref{eq:G-block}, the eigenvalues of $\BrVtG$ associated with $\Vtz_1$, $\Vtz_2$, and $\Vtz_3$ are, respectively, 
\begin{align}
    \lambda_1 = \frac{1 - \psi'}{1 + \psi'},
    \quad
    \lambda_2 = 
        \left(\frac{1 - \psi'}{1 + \psi'}\right)
        \left(\frac{1 - \psi}{1 + \psi}\right),
    \quad\text{and}\quad
    \lambda_3 = \frac{1 - \psi}{1 + \psi}. 
    \label{eq:lamb1-lamb2-lamb3}
\end{align}
We have $\omega_k = \Omega(\chi_k,\lambda_k) = 0$ if and only if $\omega_k ^\sharp \Is \Omega^\sharp(\chi_k,\lambda_k) = 0$. 
Since $\lambda_1 > \lambda_3 > \lambda_2$ and $\chi_1 = \chi_3 > \chi_2$, we see $\omega^\sharp_1 = \max_k\{\omega_k^\sharp\}$.  
Thus, North--South pattern $\Vtz_1$ (\prettyref{fig:z3}) must emerge when $\BrVtx$ becomes unstable. 
\end{proof}

\subsection{Pitchfork bifurcation from the uniform distribution}

\begin{proof}[Proof of \prettyref{prop:pitchifork}]
When $n = 2$, the uniform distribution $\BrVtx = (\Brx,\Brx)$ can be viewed as steady-state solution $y = 0$ for the following one-dimensional autonomous system:
\begin{align}
    \Dty = \Delta v(y,\mu) \Is v_1(\Vtx(y)) - v_2(\Vtx(y)), 
    \label{eq:one-dim}
\end{align}
where we choose $y \in\ClY\Is(-\Brx,\Brx)$ and let $\Vtx(y) = (x_1(y),x_2(y)) \Is (\Brx + y, \Brx - y)$. 
The bifurcation parameter $\mu$ indicates either $\phi$ or $\psi$. 
Under admissible dynamics, the bifurcation diagram becomes smoothly equivalent to the system \prettyref{eq:one-dim}. 
In the following, prime (') denotes differentiation with respect to $y$. 

In general, the system \prettyref{eq:one-dim} undergoes a pitchfork bifurcation at $(y,\mu) = (0,\mu^*)$ when $\Delta v(y,\mu)$ is odd in $y$ and the following conditions are met \citep[see, e.g., ][Section 20.1E]{Wiggins-Book2003}:
\begin{align}
    &
        \Delta v'(0,\mu^*) = 0,\quad 
        \Delta v''(0,\mu^*) = 0,\quad 
        \Delta v'''(0,\mu^*) \ne 0,\quad
        \PDF{\Delta v(0,\mu^*)}{\mu} = 0,
        \quad\text{and}\quad
        \PDF{\Delta v'(0,\mu^*)}{\mu} \ne 0. 
    \label{eq:pitchfork} 
\end{align}
We see $\Delta v$ is odd: $\Delta v(-y,\mu) = - \Delta v(y,\mu)$. The fourth condition also follows, for $\Delta v(0,\mu) = 0$ for all $\mu$. 

The first condition (\textit{nonhyperbolicity}) ensures that $\mu = \mu^*$ is the bifurcation point: 
\begin{align}
    \Delta v'(0,\mu^*) 
    & 
    = 
    \left(
    \PDF{v_1(\BrVtx)}{x_1}
    \PDF{x_1(0)}{y}
    +
    \PDF{v_1(\BrVtx)}{x_2}
    \PDF{x_2(0)}{y}
    \right)
    -
    \left(
    \PDF{v_2(\BrVtx)}{x_1}
    \PDF{x_1(0)}{y}
    +
    \PDF{v_2(\BrVtx)}{x_2}
    \PDF{x_2(0)}{y}
    \right)
    \\
    & 
    = 
    \left(
    \PDF{v_1(\BrVtx)}{x_1}
    - 
    \PDF{v_1(\BrVtx)}{x_2}
    \right)
    -
    \left(
    \PDF{v_2(\BrVtx)}{x_1}
    - 
    \PDF{v_2(\BrVtx)}{x_2}
    \right)
    = 
    \Vtz^\top \VtV_x \Vtz
    = 
    \frac{2\Brv}{\Brx}\omega(\mu^*)
    \label{eq:Htv1}
\end{align}
where $\Vtz = (1, -1)$ and we recall $\omega$ is the eigenvalue of $\VtV = \frac{\Brx}{\Brv}\VtV_x$ associated with $\Vtz$. 
The bifurcation point regarding freeness parameters $(\phi,\psi)$ is the solution for $\omega = 0$; thus, we have $\Delta v'(0,\mu) = 0$. 
For $\omega(\phi,\psi) = 0$ (or  $\Omega^\sharp(\chi(\phi),\lambda(\psi)) = 0$) to admit solution such that $(\chi,\lambda) \in (0,1)\times(0,1)$, we must have 
\begin{align}
    & \lambda^* \Is \frac{1 - \chi}{(\sigma - 1) + \sigma \chi} \in (0,1)
    & \forall \chi\in(0,1), 
    \label{eq:bif-cond}
\end{align}
or, equivalently, either $\sigma > 2$, or $\sigma\in(1,2]$ and $\chi \in (\frac{2 - \sigma}{\sigma + 1},1)$. 
We assume either of these below. 

From \prettyref{eq:Htv1}, we see
\begin{align}
    \PDF{\Delta v'(0,\mu^*)}{\mu}
    = 
    \frac{2}{\Brx}
    \left(
        \PDF{\Brv}{\mu}
        \omega(\mu^*) 
        + 
        \Brv
        \PDF{\omega(\mu^*)}{\mu} 
    \right)
    = 
    \frac{2\Brv}{\Brx}\PDF{\omega}{\mu} 
\end{align}
by noting that $\omega(\mu^*) = 0$. 
With \prettyref{eq:omega-general}, we show the fifth condition in \prettyref{eq:pitchfork}: 
\begin{align}
    &
    \PDF{\Delta v'(0,\mu^*)}{\phi}
    = 
    \frac{2\Brv}{\Brx}
    \PDF{\omega}{\chi}
    \PDF{\chi}{\phi}
    = 
    - 
    \frac{2}{\Brx}
    \cdot
    \frac
        {(1 + \lambda) (2 \sigma -1)}
        {(\sigma + (\sigma -1) \chi)^2}
    \cdot
    \frac{2}{(1 + \phi)^2} < 0, 
    \\
    &
    \PDF{\Delta v'(0,\mu^*)}{\psi}
    = 
    \frac{2\Brv}{\Brx}
    \PDF{\omega}{\lambda}
    \PDF{\lambda}{\psi}
    = 
    - 
    \frac{2}{\Brx}
    \cdot
    \frac
        {(\sigma -1) + \sigma \chi}
        {\sigma + (\sigma -1) \chi}
    \cdot
    \frac{2}{(1 + \psi)^2} < 0. 
\end{align}

For $\Delta v''(0,\mu^*)$ and $\Delta v'''(0,\mu^*)$, we resort to more explicit computations. 
Let $w_1:\ClY \to \BbR_+$ and $w_2:\ClY \to \BbR_+$ denote, respectively, the nominal wages of regions $1$ and $2$ as functions of $y\inY$. 
We first derive required derivatives of $w_1$ and $w_2$. 
By the normalization of income, we have
\begin{align}
    & 
    w_1(y) x_1(y) + w_2(y) x_2(y)
    = 
    w_1(y) (\Brx + y) + w_2(y) (\Brx - y)
    = 1 
    & \forall y\inY
\end{align}
and $w_1(0) = w_2(0) = \Brw \Is 1$. 
These imply
\begin{align}
    w_1'(0) = - w_2'(0),\quad
    w_1''(0) = w_2''(0) = - \frac{2}{\Brx} w_1'(0),\quad 
    \text{and}\quad
    w_1'''(0) = - w_2'''(0). 
    \label{eq:w_sym}
\end{align}
For instance, we have $w_1'(y)(\Brx + y) + w_1(y) + w_2'(y)(\Brx - y) - w_2'(y) = 0$. With $w_1(0) = w_2(0) = 1$, it implies that $w_1'(0) = - w_2'(0)$, which is a manifestation that the regions are symmetric when $y = 0$. 

In fact, $w_1'(0)$ is the eigenvalue of $\VtW_x(\BrVtx)$ associated with $\Vtz = (1, -1)$, 
because we have 
\begin{align}
    & 
    w_1'(0) = 
    \PDF{w_1(\BrVtx)}{x_1}\PDF{x_1(0)}{y}
    +
    \PDF{w_1(\BrVtx)}{x_2}\PDF{x_2(0)}{y}
    =
    \PDF{w_1(\BrVtx)}{x_1}
    - 
    \PDF{w_1(\BrVtx)}{x_2}
    \\
    & 
    w_2'(0) = 
    \PDF{w_2(\BrVtx)}{x_1}\PDF{x_1(0)}{y}
    +
    \PDF{w_2(\BrVtx)}{x_2}\PDF{x_2(0)}{y}
    =
    \PDF{w_2(\BrVtx)}{x_1}
    - 
    \PDF{w_2(\BrVtx)}{x_2} 
    = 
    - w_1'(0), 
\end{align}
which is, $w_1'(0) \Vtz = \VtW_x(\BrVtx)\Vtz$. 
Therefore, we have 
\begin{align}
    w_1'(0)
    = 
    - 
    \frac{\Brw}{\Brx} 
    \cdot 
    \frac{1 - (\sigma - 1)(1 + \chi) \lambda}{\Omega^\flat(\chi)},  
\end{align}
where we recall $\Omega^\flat(\chi) \Is \sigma + (\sigma - 1)\chi$. 
Then, $w''(0)$ can be evaluated by the second identity in \prettyref{eq:w_sym}. 
Additionally, by a patient algebra, we compute $w_{1}'''(0)$ as follows:
\begin{align}
    w_1'''(0) 
    = 
    \frac{2}{\Brx^3}
    \left(
        - 3 \lambda^2 
        + 
        \frac
            {3\lambda (1 + \lambda)(3 + \lambda)}
            {\Omega^\flat(\chi)}
        +
        \left(
        - 
        \frac
            {3(\sigma + 1)}
            {\Omega^\flat(\chi)^2}
        + 
        \frac
            {2\sigma (\sigma + 1)}
            {\Omega^\flat(\chi)^3}
        - 
        \frac
            {\sigma (2\sigma - 1)}
            {\Omega^\flat(\chi)^4}
        \right)
        (1 + \lambda)^3
    \right). 
\end{align}

Further, $\Vta(y) = (a_1(y),a_2(y)) \Is (a_1(\Vtx(y)),a_2(\Vtx(y)))$ satisfy
\begin{align}
    & 
    \frac{\Brx}{a_1(0)} a_1'(0) 
    = - \frac{\Brx}{a_2(0)} a_2'(0) 
    = \frac{\Brx}{\Bra} (1 - \psi) 
    = \frac{1 - \psi}{1 + \psi} 
    = \lambda
\end{align}
and $a_i''(0) = a_i'''(0) = \cdots = 0$. 
We nonte that $ a_1(0) = a_2(0) = \Bra \Is \Brx(1 + \psi)$.

By direct computations employing the above results, we confirm that 
\begin{align}
    \Delta v''(0,\mu^*) 
    = 
    \Brv (1 - \chi)
    \left( 
        - w_1'(0)^2 + w_2'(0)^2 
        + w_1''(0) - w_2''(0) 
     \right) 
     = 0 
\end{align}
from \prettyref{eq:w_sym}. 
Therefore, the second condition in \prettyref{eq:pitchfork} is met.

After some tedious calculations and manipulations, we get:
\begin{align}
    \Delta v'''(0,\mu^*) 
    = 
    - 
    \frac{4\Brv}{\Brx^3}
    \cdot
    \frac{(1 + \lambda)^3}{\Omega^\flat(\chi)^4}
    \cdot
    \Theta, 
\end{align}
where we define 
\begin{align}
    \Theta 
        \Is 
        \sigma (2\sigma - 1)
        + 
        \Omega^\flat(\chi) 
        \left(
        3
        + 
        \sigma
        \left(
            (\sigma + 1) \chi ^2 
            + 
            (4\sigma - 5) \chi +
            \sigma - 5
        \right)
        \right). 
\end{align}
We can show $\Theta > 0$ provided that bifurcation occur (i.e., condition \prettyref{eq:bif-cond} is satisfied). 
Since the other components of $\Delta v'''(0,\mu^*) $ are obviously negative, we have $\Delta v'''(0,\mu^*) < 0$ at critical values of $\phi$ or $\psi$. 

Thus, all the five conditions in  \prettyref{eq:pitchfork} are met at $\BrVtx$.  
The economy undergoes pitchfork bifurcations along smooth paths
where either the freeness of production externalities or the freeness of trade
increases, at the break points defined by $\psi^{*}$ and $\phi^{*}$,
respectively. 
Moreover, $\Delta v'''(0,\mu^*) < 0$ implies that the pitchfork bifurcations are supercritical, that is, the bifurcated branches are stable. 
\end{proof}

\section{Effects of the elasticity of substitution}
\label{app:sigma-variation}

Figures \ref{fig:stab-unif-2_full} and \ref{fig:stab-unif-4_equi_full} report, respectively, variations of Figures \ref{fig:stab-unif-2} and \ref{fig:stab-unif-4_equi} for three values of $\sigma$ which are chosen to exhaust all representative forms of the partition of the $(\phi,\psi)$-space. 
For an empirically relevant range of $\sigma$ (between $3$ and $10$), the qualitative shape of the partitions stay invariant. 
\begin{figure}[hp]
	\centering
	\begin{subfigure}[b]{.32\hsize}
	    \centering
	    \includegraphics[height=.88\hsize]{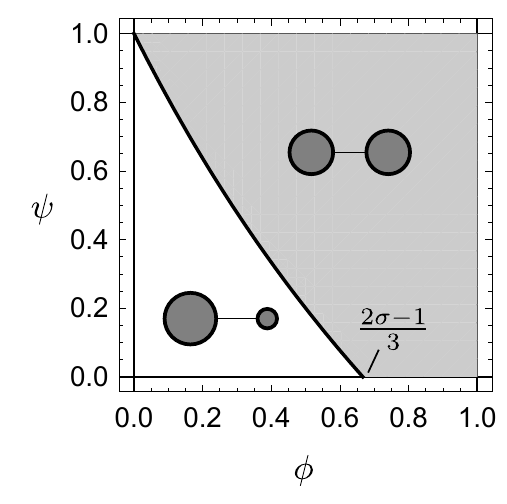}
	    \caption{$\sigma = 1.5$\label{fig:stab-unif-2-sigma1.5}} 
	\end{subfigure}
	\begin{subfigure}[b]{.32\hsize}
	    \centering
	    \includegraphics[height=.88\hsize]{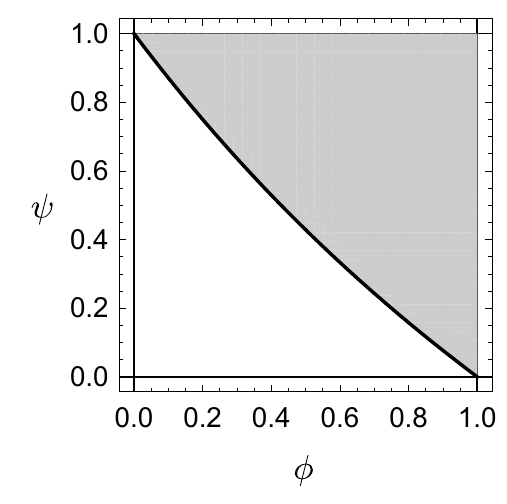}
	    \caption{$\sigma = 2.0$\label{fig:stab-unif-2-sigma2}} 
	\end{subfigure}
	\begin{subfigure}[b]{.32\hsize}
	    \centering
	    \includegraphics[height=.88\hsize]{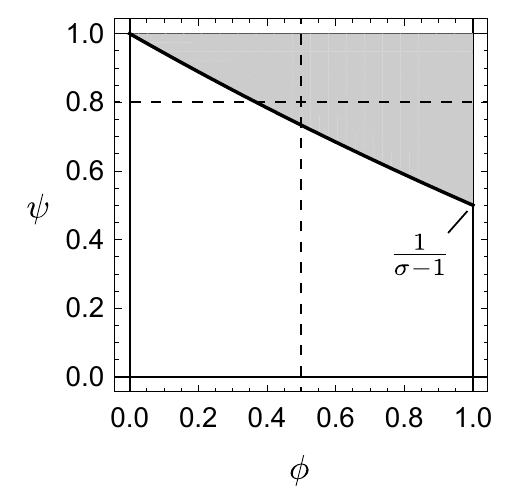}
	    \caption{$\sigma = 4.0$\label{fig:stab-unif-2-sigma4}} 
	\end{subfigure}
	\caption{Stability of $\BrVtx$ in the two-region economy.}
	\label{fig:stab-unif-2_full}
	\begin{figurenotes}
	Uniform distribution $\BrVtx$ is stable for the shaded (gray) region of $(\phi,\psi)$ and the black solid curve indicates the critical pair of $(\phi,\psi)$ where $\BrVtx$ becomes unstable. 
	The horizontal and vertical dashed lines in \prettyref{fig:stab-unif-2-sigma4} correspond to the parametric paths for the bifurcation diagrams \prettyref{fig:bif-2-phi} and \prettyref{fig:bif-2-psi}, respectively. 
	The schematic on each (gray or white) parametric region indicates the representative spatial pattern in the parametric region. 
	\end{figurenotes}

	\vskip 2em
	
	\begin{subfigure}[b]{.32\hsize}
	    \centering
	    \includegraphics[height=.88\hsize]{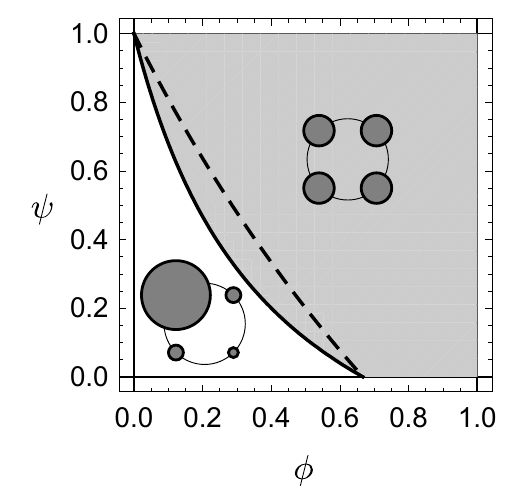}
	    \caption{$\sigma = 1.5$} 
	\end{subfigure}
	\begin{subfigure}[b]{.32\hsize}
	    \centering
	    \includegraphics[height=.88\hsize]{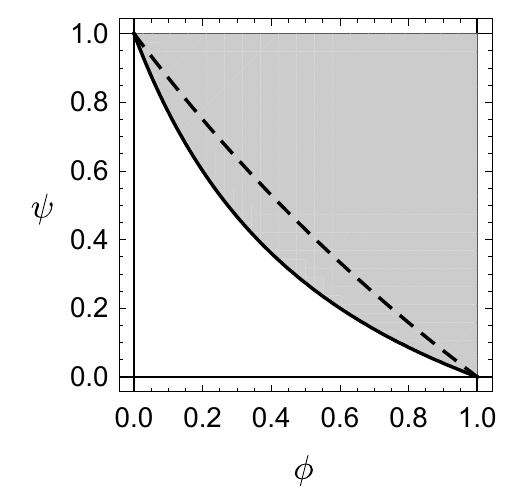}
	    \caption{$\sigma = 2.0$} 
	\end{subfigure}
	\begin{subfigure}[b]{.32\hsize}
	    \centering
	    \includegraphics[height=.88\hsize]{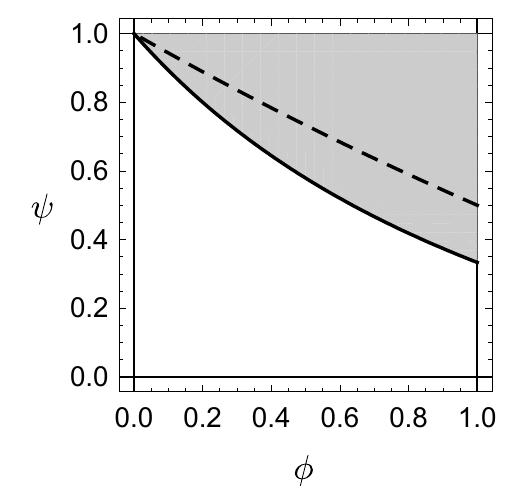}
	    \caption{$\sigma = 4.0$} 
	\end{subfigure}
	\caption{Stability of $\BrVtx$ in the four-region economies with \prettyref{eq:G-baseline} and $\VtG_e = \VtG_{g}|_{\psi' = \psi}$.}
	\begin{figurenotes}
	The black solid curves show the critical pairs $(\phi^*,\psi^*)$ at which $\BrVtx$ becomes unstable for the case with $\psi' = \psi$. The dashed curves are those for the case $\VtG = \VtG^\square$. Uniform distribution $\BrVtx$ is stable for the regions above these curves, where the gray regions correspond to $\VtG = \VtG_e$. Observe that the solid curves always stay below the dashed curves. That is, $\BrVtx$ is stable for a wider range of $(\phi,\psi)$ when the economy is more connected. 
	\end{figurenotes}
	\label{fig:stab-unif-4_equi_full}
\end{figure}

\clearpage 

{
\small
\singlespacing
\bibliographystyle{apalike}
\bibliography{refs}

\begin{thebibliography}{}

\bibitem[Ahlfeldt et~al., 2015]{Ahlfeldt-et-al-ECTA2015}
Ahlfeldt, G.~M., Redding, S.~J., Sturm, D.~M., and Wolf, N. (2015).
\newblock The economics of density: Evidence from the {B}erlin {W}all.
\newblock {\em Econometrica}, 83(6):2127--2189.

\bibitem[Aizawa et~al., 2020]{Aizawa-etal-IJBC2020}
Aizawa, H., Ikeda, K., Osawa, M., and Gaspar, J. (2020).
\newblock Breaking and sustaining bifurcations in $s_n$-invariant equidistant
  economy.
\newblock {\em International Journal of Bifurcation and Chaos}, 30(16):2050240.

\bibitem[Akamatsu et~al., 2019]{AMOT-DP2019}
Akamatsu, T., Mori, T., Osawa, M., and Takayama, Y. (2019).
\newblock Endogenous agglomeration in a many-region world.
\newblock \url{https://arxiv.org/abs/1912.05113}.

\bibitem[Akamatsu et~al., 2012]{Akamatsu-Takayama-Ikeda-JEDC2012}
Akamatsu, T., Takayama, Y., and Ikeda, K. (2012).
\newblock Spatial discounting, {F}ourier, and racetrack economy: A recipe for
  the analysis of spatial agglomeration models.
\newblock {\em Journal of Economic Dynamics and Control}, 99(11):32--52.

\bibitem[Allen and Arkolakis, 2014]{Allen-Arkolakis-QJE2014}
Allen, T. and Arkolakis, C. (2014).
\newblock Trade and the topography of the spatial economy.
\newblock {\em The Quarterly Journal of Economics}, 129(3):1085--1140.

\bibitem[Anderson and {v}an Wincoop, 2004]{Anderson-Wincoop-JEL2004}
Anderson, J.~E. and {v}an Wincoop, E. (2004).
\newblock Trade costs.
\newblock {\em Journal of Economic Literature}, 42(3):691--751.

\bibitem[Armington, 1969]{Armington-IMF1969}
Armington, P.~S. (1969).
\newblock A theory of demand for product distinguished by place of production.
\newblock {\em International Monetary Fund Staff Papers}, 16(1):159--178.

\bibitem[Baldwin, 2016]{Baldwin-Book2016}
Baldwin, R. (2016).
\newblock {\em The Great Convergence}.
\newblock Harvard University Press.

\bibitem[Barbero and Zof{\'\i}o, 2016]{Barbero-Zofio-NETS2016}
Barbero, J. and Zof{\'\i}o, J.~L. (2016).
\newblock The multiregional core-periphery model: The role of the spatial
  topology.
\newblock {\em Networks and Spatial Economics}, 16(2):469--496.

\bibitem[Baum-Snow, 2007]{Baum-Snow-QJE2007}
Baum-Snow, N. (2007).
\newblock Did highways cause suburbanization?
\newblock {\em The Quarterly Journal of Economics}, 122(2):775--805.

\bibitem[Baum-Snow et~al., 2017]{Baum-Snow-et-al-REStat2017}
Baum-Snow, N., Brandt, L., Henderson, J.~V., Turner, M.~A., and Zhang, Q.
  (2017).
\newblock Roads, railroads, and decentralization of {C}hinese cities.
\newblock {\em Review of Economics and Statistics}, 99(3):435--448.

\bibitem[Beckmann, 1976]{Beckmann-Book1976}
Beckmann, M.~J. (1976).
\newblock Spatial equilibrium in the dispersed city.
\newblock In Papageorgiou, Y.~Y., editor, {\em Mathematical Land Use Theory}.
  Lexington Book.

\bibitem[Behrens et~al., 2009]{Behrens-etal-JUE2009}
Behrens, K., Gaign{\'e}, C., and Thisse, J.-F. (2009).
\newblock Industry location and welfare when transport costs are endogenous.
\newblock {\em Journal of Urban Economics}, 65(2):195--208.

\bibitem[Behrens and Picard, 2011]{Behrens-Picard-JIE2011}
Behrens, K. and Picard, P.~M. (2011).
\newblock Transportation, freight rates, and economic geography.
\newblock {\em Journal of International Economics}, 85(2):280--291.

\bibitem[Berliant and Fujita, 2012]{Berliant-Fujita-RSUE2012}
Berliant, M. and Fujita, M. (2012).
\newblock Culture and diversity in knowledge creation.
\newblock {\em Regional Science and Urban Economics}, 42(4):648--662.

\bibitem[Berliant and Kung, 2009]{Berliant-Kung-RSUE2009}
Berliant, M. and Kung, F.-c. (2009).
\newblock Bifurcations in regional migration dynamics.
\newblock {\em Regional Science and Urban Economics}, 39(6):714--720.

\bibitem[Cairncross, 1997]{cairncross1997death}
Cairncross, F. (1997).
\newblock {\em The Death of Distance: How the Communications Revolution Will
  Change Our Lives}, volume 302.
\newblock Harvard Business School Press Boston, MA.

\bibitem[Castro et~al., 2021]{Castro_2021}
Castro, S. B. S.~D., Correia-da Silva, J., and Gaspar, J.~M. (2021).
\newblock Economic geography meets hotelling: the home-sweet-home effect.
\newblock {\em Economic Theory}.

\bibitem[Duranton and Puga, 2001]{Duranton-Puga-AER2001}
Duranton, G. and Puga, D. (2001).
\newblock Nursery cities: {U}rban diversity, process innovation, and the life
  cycle of products.
\newblock {\em American Economic Review}, 91(5):1454--1477.

\bibitem[Duranton and Puga, 2004]{Duranton-Puga-HB2004}
Duranton, G. and Puga, D. (2004).
\newblock Micro-foundations of urban agglomeration economies.
\newblock In Henderson, J.~V. and Thisse, J.-F., editors, {\em Handbook of
  Regional and Urban Economics}, volume~4, pages 2063--2117. North-Holland.

\bibitem[Duranton and Puga, 2015]{Duranton-Puga-HB2015}
Duranton, G. and Puga, D. (2015).
\newblock Urban land use.
\newblock In Duranton, G., Henderson, J.~V., and Strange, W.~C., editors, {\em
  Handbook of Regional and Urban Economics}, volume~5, pages 467--560.
  Elsevier.

\bibitem[Facchinei and Pang, 2003]{Facchinei-Pang-Book2003}
Facchinei, F. and Pang, J.-S. (2003).
\newblock {\em Finite-dimensional Variational Inequalities and Complementarity
  Problems}.
\newblock Springer Science \& Business Media.

\bibitem[Fujita, 2007]{Fujita-RSUE2007}
Fujita, M. (2007).
\newblock Towards the new economic geography in the brain power society.
\newblock {\em Regional Science and Urban Economics}, 37(4):482--490.

\bibitem[Fujita et~al., 1999]{Fujita-Krugman-Venables-Book1999}
Fujita, M., Krugman, P., and Venables, A. (1999).
\newblock {\em The Spatial Economy: Cities, Regions, and International Trade}.
\newblock Princeton University Press.

\bibitem[Fujita and Mori, 2005]{Fujita-Mori-PRS2005}
Fujita, M. and Mori, T. (2005).
\newblock Frontiers of the new economic geography.
\newblock {\em Papers in Regional Science}, 84(3):377--405.

\bibitem[Fujita and Ogawa, 1982]{Fujita-Ogawa-RSUE1982}
Fujita, M. and Ogawa, H. (1982).
\newblock Multiple equilibria and structural transition of non-monocentric
  urban configurations.
\newblock {\em Regional Science and Urban Economics}, 12:161--196.

\bibitem[Fujita and Thisse, 2013]{Fujita-Thisse-Book2013}
Fujita, M. and Thisse, J.-F. (2013).
\newblock {\em Economics of Agglomeration: Cities, Industrial Location, and
  Regional Growth (2nd Edition)}.
\newblock Cambridge University Press.

\bibitem[Gaspar, 2018]{gaspar2018prospective}
Gaspar, J.~M. (2018).
\newblock A prospective review on new economic geography.
\newblock {\em The Annals of Regional Science}, 61(2):237--272.

\bibitem[Gaspar, 2020a]{gaspar2020new}
Gaspar, J.~M. (2020a).
\newblock New economic geography: Economic integration and spatial imbalances.
\newblock In {\em Spatial Economics Volume I}, pages 79--110. Springer.

\bibitem[Gaspar, 2020b]{gaspar2020history}
Gaspar, J.~M. (2020b).
\newblock New economic geography: history and debate.
\newblock {\em The European Journal of the History of Economic Thought}, pages
  1--37.

\bibitem[Gaspar et~al., 2018]{Gaspar-et-al-ET2018}
Gaspar, J.~M., Castro, S.~B., and Correia-da Silva, J. (2018).
\newblock Agglomeration patterns in a multi-regional economy without income
  effects.
\newblock {\em Economic Theory}, 66(4):863--899.

\bibitem[Gaspar et~al., 2019]{Gaspar-et-al-IJET2019}
Gaspar, J.~M., Castro, S.~B., and Correia-da Silva, J. (2019).
\newblock The footloose entrepreneur model with a finite number of equidistant
  regions.
\newblock {\em International Journal of Economic Theory}, 16(4):420--446.

\bibitem[Gaspar et~al., 2021]{Gaspar-et-al-RSUE2021}
Gaspar, J.~M., Ikeda, K., and Onda, M. (2021).
\newblock Global bifurcation mechanism and local stability of identical and
  equidistant regions: Application to three regions and more.
\newblock {\em Regional Science and Urban Economics}, 86:103597.

\bibitem[Helpman, 1998]{Helpman-Book1998}
Helpman, E. (1998).
\newblock The size of regions.
\newblock In Pines, D., Sadka, E., and Zilcha, I., editors, {\em Topics in
  Public Economics: Theoretical and Applied Analysis}, pages 33--54. Cambridge
  University Press.

\bibitem[Helsley and Strange, 2014]{Helsley-Strange-JPE2014}
Helsley, R.~W. and Strange, W.~C. (2014).
\newblock Coagglomeration, clusters, and the scale and composition of cities.
\newblock {\em Journal of Political Economy}, 122(5):1064--1093.

\bibitem[Helsley and Zenou, 2014]{Helsley-Zenou-JET2014}
Helsley, R.~W. and Zenou, Y. (2014).
\newblock Social networks and interactions in cities.
\newblock {\em Journal of Economic Theory}, 150:426--466.

\bibitem[Horn and Johnson, 2012]{horn2012matrix}
Horn, R.~A. and Johnson, C.~R. (2012).
\newblock {\em Matrix Analysis}.
\newblock Cambridge University Press.

\bibitem[Ikeda et~al., 2021]{RePEc:cap:wpaper:012021}
Ikeda, K., Aizawa, H., and Gaspar, J.~M. (2021).
\newblock {How and where satellite cities form around a large city: Bifurcation
  mechanism of a long narrow economy}.
\newblock Working Papers de Economia (Economics Working Papers)~01,
  Cat^^c3^^b3lica Porto Business School, Universidade Cat^^c3^^b3lica
  Portuguesa.

\bibitem[Ikeda et~al., 2012]{ikeda2012spatial}
Ikeda, K., Akamatsu, T., and Kono, T. (2012).
\newblock Spatial period-doubling agglomeration of a core--periphery model with
  a system of cities.
\newblock {\em Journal of Economic Dynamics and Control}, 36(5):754--778.

\bibitem[Jackson, 2010]{jackson2010social}
Jackson, M.~O. (2010).
\newblock {\em Social and Economic Networks}.
\newblock Princeton University Press.

\bibitem[Jonkeren et~al., 2011]{Jonkeren-etal-JEG2011}
Jonkeren, O., Demirel, E., van Ommeren, J., and Rietveld, P. (2011).
\newblock Endogenous transport prices and trade imbalances.
\newblock {\em Journal of Economic Geography}, 11(3):509--527.

\bibitem[Krugman, 1991]{Krugman-JPE1991}
Krugman, P.~R. (1991).
\newblock Increasing returns and economic geography.
\newblock {\em Journal of Political Economy}, 99(3):483--499.

\bibitem[Krugman, 1993]{Krugman-EER1993}
Krugman, P.~R. (1993).
\newblock On the number and location of cities.
\newblock {\em European Economic Review}, 37(2):293--298.

\bibitem[Krugman and Venables, 1995]{Krugman-Venables-QJE1995}
Krugman, P.~R. and Venables, A.~J. (1995).
\newblock Globalization and the inequality of nations.
\newblock {\em The Quarterly Journal of Economics}, 110(4):857--880.

\bibitem[Kuznetsov, 2004]{kuznetsov2013elements}
Kuznetsov, Y.~A. (2004).
\newblock {\em Elements of Applied Bifurcation Theory (3rd Eds.)}.
\newblock Springer-Verlag.

\bibitem[Lucas and Rossi-Hansberg, 2002]{Lucas-Rossi-Hansberg-ECTA2002}
Lucas, R.~E. and Rossi-Hansberg, E. (2002).
\newblock On the internal structure of cities.
\newblock {\em Econometrica}, 70(4):1445--1476.

\bibitem[Mas-Colell et~al., 1995]{MWG}
Mas-Colell, A., Whinston, M.~D., Green, J.~R., et~al. (1995).
\newblock {\em Microeconomic Theory}, volume~1.
\newblock Oxford University Press.

\bibitem[Matsuyama, 2017]{matsuyama2017geographical}
Matsuyama, K. (2017).
\newblock Geographical advantage: Home market effect in a multi-region world.
\newblock {\em Research in Economics}, 71(4):740--758.

\bibitem[Murata, 2003]{Murata-JUE2003}
Murata, Y. (2003).
\newblock Product diversity, taste heterogeneity, and geographic distribution
  of economic activities:: market vs. non-market interactions.
\newblock {\em Journal of Urban Economics}, 53(1):126--144.

\bibitem[Osawa and Akamatsu, 2020]{Osawa-Akamatsu-2020}
Osawa, M. and Akamatsu, T. (2020).
\newblock Equilibrium refinement for a model of non-monocentric internal
  structures of cities: A potential game approach.
\newblock {\em Journal of Economic Theory}, 187:105025.

\bibitem[Ottaviano and Peri, 2006]{Ottaviano-Peri-JEG2006}
Ottaviano, G.~I. and Peri, G. (2006).
\newblock The economic value of cultural diversity: {E}vidence from {US}
  cities.
\newblock {\em Journal of Economic Geography}, 6(1):9--44.

\bibitem[Ottaviano and Peri, 2008]{Ottaviano-Peri-DP2008}
Ottaviano, G.~I. and Peri, G. (2008).
\newblock Immigration and national wages: {C}larifying the theory and the
  empirics.
\newblock Technical report, National Bureau of Economic Research.

\bibitem[Ottaviano and Prarolo, 2009]{Ottaviano-Prarolo-JRS2009}
Ottaviano, G.~I. and Prarolo, G. (2009).
\newblock Cultural identity and knowledge creation in cosmopolitan cities.
\newblock {\em Journal of Regional Science}, 49(4):647--662.

\bibitem[Picard and Zenou, 2018]{Picard-Zenou-JUE2018}
Picard, P.~M. and Zenou, Y. (2018).
\newblock Urban spatial structure, employment and social ties.
\newblock {\em Journal of Urban Economics}, 104:77--93.

\bibitem[Proost and Thisse, 2019]{Proost-Thisse-JEL2019}
Proost, S. and Thisse, J.-F. (2019).
\newblock What can be learned from spatial economics?
\newblock {\em Journal of Economic Literature}, 57(3):575--643.

\bibitem[Rosenthal and Strange, 2020]{Rosenthal-Strange-JEP2020}
Rosenthal, S.~S. and Strange, W.~C. (2020).
\newblock How close is close? {T}he spatial reach of agglomeration economies.
\newblock {\em Journal of Economic Perspectives}, 34(3):27--49.
\newblock Unpublished manuscript.

\bibitem[Sandholm, 2010]{Sandholm-Book2010}
Sandholm, W.~H. (2010).
\newblock {\em Population Games and Evolutionary Dynamics}.
\newblock MIT Press.

\bibitem[Tabuchi, 1998]{Tabuchi-JUE1998}
Tabuchi, T. (1998).
\newblock Urban agglomeration and dispersion: A synthesis of alonso and
  krugman.
\newblock {\em Journal of Urban Economics}, 44(3):333--351.

\bibitem[Tabuchi, 2014]{Tabuchi-RSUE2014}
Tabuchi, T. (2014).
\newblock Historical trends of agglomeration to the capital region and new
  economic geography.
\newblock {\em Regional Science and Urban Economics}, 44:50--59.

\bibitem[Tabuchi and Thisse, 2011]{Tabuchi-Thisse-JUE2011}
Tabuchi, T. and Thisse, J.-F. (2011).
\newblock A new economic geography model of central places.
\newblock {\em Journal of Urban Economics}, 69(2):240--252.

\bibitem[Venables, 1996]{venables1996equilibrium}
Venables, A.~J. (1996).
\newblock Equilibrium locations of vertically linked industries.
\newblock {\em International Economic Review}, pages 341--359.

\bibitem[Wiggins, 2003]{Wiggins-Book2003}
Wiggins, S. (2003).
\newblock {\em Introduction to Applied Nonlinear Dynamical Systems and Chaos}.
\newblock Springer Science \& Business Media.

\end{thebibliography}
}

\end{document}